\documentclass [10pt,article,oneside]{memoir}
\usepackage{mystandard}
\acinq
\aquatre

\usepackage[english]{babel}
\usepackage[utf8]{inputenc}
\usepackage{graphicx,csquotes,algorithm}
\usepackage{algpseudocode}

\usepackage[all]{xy}
\xyoption{poly}
\xyoption{arrow}
\SelectTips{xy}{12}

\bibliography{biblio.bib}

\newcommand{\V}{\mathcal{V}}
\newcommand{\ew}{\varepsilon}

\newcommand{\iid}{\textsl{i.i.d.}}
\newcommand{\B}{\partial}

\newcommand{\Hbar}{\overline{H}}
\newcommand{\Ybar}{\overline{Y}}
\newcommand{\Zbar}{\overline{Z}}

\renewcommand{\H}{\mathcal{H}}
\newcommand{\xbar}{\overline{x}}
\newcommand{\BSA}{{\normalfont\textsf{PSA}}}
\newcommand{\FSA}{{\normalfont\textsf{PFSA}}}

\newcommand{\unit}{\mathbf{e}}

\newcommand{\EOF}{\textsc{eof}}
\newcommand{\DL}{\textsc{dl}}
\newcommand{\WFI}{\textsc{wfi}}

\begin{document}

\title{\huge\bfseries Synchronization of Bernoulli \\sequences on shared letters}
\author{Samy Abbes
    \\
 Universit\'e Paris Diderot/IRIF CNRS UMR 8243\\
    \normalsize\texttt{samy.abbes@univ-paris-diderot.fr}}
\date{November 2016}
\maketitle

\begin{abstract}
  The topic of this paper is the distributed and incremental
  generation of long executions of concurrent systems, uniformly or
  more generally with weights associated to elementary actions.

  Synchronizing sequences of letters on alphabets sharing letters are
  known to produce a trace in the concurrency theoretic sense,
  \emph{i.e.}, a labeled partially ordered set. We study the probabilistic
  aspects by considering the synchronization of Bernoulli sequences of
  letters, under the light of Bernoulli and uniform measures recently
  introduced for trace monoids.

  We introduce two algorithms that produce random traces, using only
  local random primitives.  We thoroughly study some specific
  examples, the path model and the ring model, both of arbitrary
  size. For these models, we show how to generate any Bernoulli
  distributed random traces, which includes the case of uniform
  generation.
\end{abstract}

\section{Introduction}
\label{sec:introduction}

The developments of concurrency theory and of model checking theory
have urged the development of a theory of probabilistic concurrent
systems. A central issue is the problem of the uniform generation of
long executions of concurrent systems. Executions of a concurrent
system can be represented as partial orders of events. Each partial
order has several sequentializations, the combinatorics of which is
non trivial. Therefore, the uniform generation of executions of a
concurrent system is much different from the uniform generation of
their sequentializations. The later can be done using uniform
generation techniques for runs of transition systems, at the expense
of an increasing amount of complexity due to the combinatorics of
sequentializations~\cite{oudinet07:_unifor,denise06:_unifor,sen07:_effec}.
Yet, it still misses the overall goal of uniform generation among
executions of the system, since it focuses on their
sequentializations.


We consider the framework of \emph{trace monoids}~\cite{diekert90},
also called monoids with partial commutations~\cite{cartier69}. Trace
monoids have been studied as basic models of concurrent systems since
several decades~\cite{mazurkiewicz87:_trace,diekert95}. One uses the
algebraic commutation of generators of the monoid to render the
concurrency of elementary actions. The elements of the trace monoid,
called \emph{traces}, represent the finite executions of the
concurrent system~\cite{diekert90}. More sophisticated concurrency
models build on trace monoids, for instance executions of $1$-safe
Petri nets correspond to regular languages of trace
monoids~\cite{reisig98:_petri_i}.

The topic of this paper is the effective uniform generation of large
traces, \emph{i.e.}, large elements in a trace monoid. It has several
potential applications in the model checking and in the simulation of
concurrent systems.

In a recent work co-authored with J.~Mairesse~\cite{abbes_mair15}, we
have shown that the notion of Bernoulli measure for trace monoids
provides an analogous, in a framework with concurrency, of classical
Bernoulli sequences---\emph{i.e.}, the mathematical model of
memoryless coin tossing. In particular, Bernoulli measures encompass
the \emph{uniform measure}, an analogous for trace monoids of the
maximal entropy measure. Therefore Bernoulli measures are an adequate
theoretical ground to work with for the random generation of traces,
in particular for the uniform generation.

Bernoulli sequences are highly efficiently approximated by random
generators. An obvious, but nevertheless crucial feature of their
generation is that it is incremental. For the random generation
of large traces, we shall also insist that the generation procedure is
incremental. Furthermore, another desirable feature is that it is
distributed, in a sense that we explain now.

We consider trace monoids attached to networks of alphabets sharing
common letters.  The synchronization of several sequences of letters
on different local alphabets sharing common letters is known to be
entirely encoded by a unique element of a trace monoid. If $\Sigma$
denotes the union of all local alphabets, then the
\emph{synchronization trace monoid} is the monoid with the
presentation by generators and relations
$\M= \langle\Sigma\;|\; ab=ba \rangle$, where $(a,b)$ ranges over
pairs of letters that do not occur in any common local
alphabet. Hence, seeing local alphabets as ``resources'', two distinct
letters $a$ and $b$ commute in $\M$ if and only if they do not share
any common resource---a standard paradigm in concurrency theory.

In this framework, our problem rephrases as follows: given a
synchronization trace monoid, design a probabilistic protocol to
reconstruct a global random trace, uniformly among traces, and in a
distributed and incremental way. By ``distributed'', we mean that the
random primitives should only deal with the local alphabets. The
expression ``uniformly among traces'' deserves also an explanation,
since traces of a monoid are countably many. One interpretation is to
fix a size $k$ for target traces, and to retrieve a trace uniformly
distributed among those of size~$k$. Another interpretation is to
consider infinite traces, \emph{i.e.}, endless executions of the
concurrent system. It amounts in an idealization of the case with
large size traces.  We then rely on the notion of uniform measure for
infinite traces, which happens to have nicer properties than the
uniform distribution on traces of fixed size.  As explained above, the
uniform measure belongs to the largest class of Bernoulli measures for
trace monoids. Hence, a slightly more general problem is the
distributed and incremental reconstruction of any Bernoulli measure
attached to a synchronization trace monoid.

We introduce two algorithms that partly solve this problem, the
Probabilistic Synchronization Algorithm (\BSA) and the Probabilistic
Full Synchronization Algorithm (\FSA), the later building on the
former. According to the topology of the network, one algorithm or the
other shall be applied.  We show that the problem of generating traces
according to a Bernoulli measure is entirely solved for some specific
topologies of the network, namely for the path topology and for the
ring topology---it could also be applied successfully to a tree
topology. It is only partly solved for a general topology.  Yet, even
in the case of a general topology, our procedure outputs large random
traces according to a Bernoulli scheme, although it is unclear how to
tune the probabilistic parameters in order to obtain
uniformity. Furthermore, the amount of time needed to obtain a trace
of size $k$ is linear with $k$ in average.

Several works analyzing the exchange of information in concurrent
systems restrict their studies to tree topologies, see for
instance~\cite{genest13:_async}, and very few has been said on the
probabilistic aspects.  In particular, designing an incremental and
distributed procedure to uniformly randomize a system with a ring
topology has been an unsolved problem in the literature so far.

How does our method compare with standard generation methods based on
tools from Analytic Combinatorics, such as Boltzmann samplers
techniques? There exists a normal form for traces, the so-called
Cartier-Foata normal form, from which one derives a bijection between
traces of a given trace monoid and words of a regular language. This
seems to draw a direct connection with Boltzmann sampling of words
from regular languages
(\cite[p.590]{duchon04},~\cite{bernardi12}). This approach however
suffers from some drawbacks. First, the rejection mechanism which is
at the very heart of the Boltzmann sampling approach prevents the
construction to be incremental, as we seek. This could be avoided by
considering instead the generation of the Markov chain of the elements
of the normal form associated to the uniform measure on infinite
traces, as investigated partly in~\cite{abbes15_unif}. But this leads
to a second problem, namely that the randomness now concentrates on
the set of cliques of the trace monoid, a set which size grows
exponentially fast with the number of generators of the monoid in
general. It henceforth misses the point of being a distributed
generation.

\medskip \emph{Outline.}\quad Section~\ref{sec:illustrating-bsa-fsa}
illustrates on small examples the two algorithms for generating random
traces from a network of alphabets, the \BSA\ and the \FSA, to be
fully analyzed in forthcoming
sections. Sections~\ref{sec:prel-trace-mono} and
Section~\ref{sec:prel-prob} are two preliminary sections, gathering
material on the combinatorics of trace monoids for the first one, and
material on Bernoulli measures for trace monoids for the second
one. The non-random algorithms for the synchronization of sequences,
possibly infinite, that we present in
Section~\ref{sec:prel-trace-mono} seem to be new, although the
fundamental ideas on which they rest are not new. A new result on the
theory of M\"obius valuations is given at the end of
Section~\ref{sec:prel-prob}.

Our main contributions are organized in
Sections~\ref{sec:basic-synchr-algor}
and~\ref{sec:iter-meas-full}. Section~\ref{sec:basic-synchr-algor} is
devoted to the Probabilistic Synchronization Algorithm (\BSA). The
analysis of the algorithm itself is quite simple and short. The most
striking contributions are in the examples. In particular, for the
surprising case of the path model, the \BSA\ is shown to work the best
one could expect.

Section~\ref{sec:iter-meas-full} is devoted to the description and to
the analysis of the Probabilistic Full Synchronization Algorithm
(\FSA), both from the probabilistic point of view and from the
complexity point of view. Examples are examined; the ring model
of arbitrary size is precisely studied, and we obtain the satisfying
result of simulating any Bernoulli scheme on infinite traces by means
of the \FSA.

Finally, Section~\ref{sec:computational-issues} discusses some
additional complexity issues and suggests perspectives.

\section{Illustrating the \BSA\ and \FSA\ algorithms}
\label{sec:illustrating-bsa-fsa}

In this section, we illustrate our two generation algorithms on small
examples, as well as the tuning of their probabilistic parameters, at
an informal and descriptive level.

\subsection{Illustrating the Probabilistic Synchronization Algorithm (1)}
\label{sec:bsa-algorithm}

Let $a_0,a_1,a_2,a_3,a_4$ be five distinct symbols, and consider the
four alphabets:
\begin{align*}
  \Sigma_1&=\{a_0,a_1\},&
  \Sigma_2&=\{a_1,a_2\},&
  \Sigma_3&=\{a_2,a_3\},&
  \Sigma_4&=\{a_3,a_4\}.
\end{align*}

To each alphabet $\Sigma_i$ is attached a device able to produce a
random sequence
$Y_i=(Y_{i,1},Y_{i,2},\dots)$ of letters $Y_{i,j}\in\Sigma_i$. The random
letters are independent and identically distributed, according to a
distribution $p_i=\left(\xymatrix@1@C=1em{
p_i(a_{i-1})&p_i(a_i)}\right)$
to be determined later, but with positive coefficients. Furthermore,
we assume that the devices themselves are probabilistically
independent with respect to each other.  For instance, the beginnings
of the four sequences might be:
\begin{align*}
  Y_1&=(a_1a_0a_0a_1a_0\dots)&Y_2&=(a_1a_2a_2a_2a_2\dots)\\
  Y_3&=(a_2a_2a_3a_3a_2\dots)&Y_4&=(a_4a_3a_3a_3a_4\dots)
\end{align*}

Then we stack these four sequences in a unique vector $Y$ with four
coordinates, and we join the matching letters from one coordinate to
its neighbor coordinates, in their order of appearance. The border
elements $a_0$ and $a_4$ do not match with any other element; some
elements are not yet matched (for example, the second occurrence of
$a_1$ in~$Y_1$), and we do not picture them. This yields:

\begin{gather*}
  Y=\left(\begin{gathered}
\xymatrix@R=.5em@C=1ex{a_1\ar@{-}[d]&a_0&a_0&&&&\dots\\
a_1&a_2\ar@{-}[d]&a_2\ar@{-}[d]&&&a_2\ar@{-}[d]&\dots\\
&a_2&a_2&a_3\ar@{-}[d]&a_3\ar@{-}[d]&a_2&\dots\\
a_4&&&a_3&a_3&&\dots
}
\end{gathered}
\right)
\end{gather*}

Finally, we identify each matching pair with a single \emph{piece},
labeled with the matching letter and we impose a rotation of the
whole picture by a quarter turn counterclockwise. We obtain the
\emph{heap of pieces} depicted in Figure~\ref{fig:heapsosj}. It is
apparent on this picture that pieces $a_i$ and $a_{i+1}$ share a sort
of common resource. Therefore we depict the topology of the system
resulting from the synchronization of the network of alphabets
$(\Sigma_1,\Sigma_2,\Sigma_3,\Sigma_4)$ in Figure~\ref{fig:poqwpqq},
by drawing an undirected graph with pieces as vertices, and an edge
between two pieces whenever they share a common resource, or
equivalently, whenever they both appear in the same alphabet.

The \BSA\ algorithm consists, for each device attached to each
alphabet~$\Sigma_i$, to do the following:
\begin{inparaenum}[1)]
  \item Generate its own local random sequence;
  \item Communicate with the neighbors in order to tag the matching
    of each of the occurring letters.
\end{inparaenum}

\begin{figure}
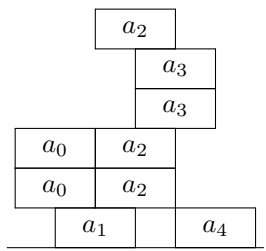

\begin{gather*}
  \xy<.15em,0em>:
(-2,0);(62,0)**@{-},
(10,0)="G",
@={"G","G"+(20,0),"G"+(20,10),"G"+(0,10)},
s0="prev" @@{;"prev";**@{-}="prev"}, "G"+(10,5)*{a_1},
(40,0)="G",
@={"G","G"+(20,0),"G"+(20,10),"G"+(0,10)},
s0="prev" @@{;"prev";**@{-}="prev"}, "G"+(10,5)*{a_4},
(0,10)="G",
@={"G","G"+(20,0),"G"+(20,10),"G"+(0,10)},
s0="prev" @@{;"prev";**@{-}="prev"}, "G"+(10,5)*{a_0},
(0,20)="G",
@={"G","G"+(20,0),"G"+(20,10),"G"+(0,10)},
s0="prev" @@{;"prev";**@{-}="prev"}, "G"+(10,5)*{a_0},
(20,10)="G",
@={"G","G"+(20,0),"G"+(20,10),"G"+(0,10)},
s0="prev" @@{;"prev";**@{-}="prev"}, "G"+(10,5)*{a_2},
(20,20)="G",
@={"G","G"+(20,0),"G"+(20,10),"G"+(0,10)},
s0="prev" @@{;"prev";**@{-}="prev"}, "G"+(10,5)*{a_2},
(30,30)="G",
@={"G","G"+(20,0),"G"+(20,10),"G"+(0,10)},
s0="prev" @@{;"prev";**@{-}="prev"}, "G"+(10,5)*{a_3},
(30,40)="G",
@={"G","G"+(20,0),"G"+(20,10),"G"+(0,10)},
s0="prev" @@{;"prev";**@{-}="prev"}, "G"+(10,5)*{a_3},
(20,50)="G",
@={"G","G"+(20,0),"G"+(20,10),"G"+(0,10)},
s0="prev" @@{;"prev";**@{-}="prev"}, "G"+(10,5)*{a_2},
\endxy
\end{gather*}
\caption{Heap of pieces corresponding to the vector $Y$}
  \label{fig:heapsosj}
\end{figure}

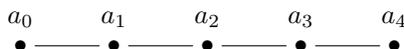
\begin{figure}
  \begin{gather*}
    \xymatrix{\bullet\ar@{-}[r]\POS!U(2)\drop{a_0}&
\bullet\ar@{-}[r]\POS!U(2)\drop{a_1}&
\bullet\ar@{-}[r]\POS!U(2)\drop{a_2}&
\bullet\ar@{-}[r]\POS!U(2)\drop{a_3}&
\bullet\POS!U(2)\drop{a_4}}
  \end{gather*}
  \caption{Synchronization graph of the path model with five generators}
  \label{fig:poqwpqq}
\end{figure}
Since the coefficients $p_i(a_{i-1})$ and $p_i(a_i)$ are positive,
each coordinate $Y_i$ has infinitely many occurrences of both letters
$a_{i-1}$ and~$a_i$. Furthermore, each occurrence of a letter $a_j$
with $j\neq0,4$ will eventually match a corresponding occurrence in a
neighbor coordinate. Therefore the random heap that we obtain,
say~$\xi$, is infinite if the algorithm keeps running forever. The
\BSA\ produces incremental finite approximations of~$\xi$. The
probabilistic analysis of the \BSA\ consists in determining the
probability distribution of the theoretical infinite random heap $\xi$
that the \BSA\ would produce if it was to run indefinitely.

In order to characterize the probability distribution of~$\xi$, we
shall denote by~$\up x$, for each finite possible heap~$x$, the
probabilistic event that $\xi$ \emph{starts with~$x$}. This means that
the finite heap $x$ can be seen at the bottom of~$\xi$. Equivalently:
$\xi\in\up x$ if, after waiting long enough, the finite heap $x$ can
be seen at the bottom of the current heap produced by the \BSA.

We aim at evaluating the probability $\pr(\up x)$ for every finite
heap~$x$. For this, consider an arbitrary possible sequentialization
of~$x$, seen as a successive piling of different occurrences of the
elementary pieces, which we write symbolically as
$x=x_1\cdot\ldots\cdot x_k$ with $x_j\in\{a_0,a_1,a_2,a_3,a_4\}$.
Then $\pr(\up x)$ is obtained as the following product:
\begin{gather*}
  \pr(\up x)=t_{x_1}\cdot\ldots\cdot t_{x_k},
\end{gather*}
where the coefficient $t_{x_j}$ corresponds to the probability
of having the letter $x_j$ appearing either on its single coordinate
if $x_j=a_0$ or $x_j=a_4$, or on both coordinates to which it belongs
otherwise. Hence:
\begin{align*}
  t_{a_0}&=p_1(a_0)&t_{a_1}&=p_1(a_1)\cdot
                             p_2(a_1)&t_{a_2}&=p_2(a_2)\cdot p_3(a_2)\\
  t_{a_3}&=p_3(a_3)\cdot
           p_4(a_3)& t_{a_4}&=p_4(a_4).
\end{align*}

For instance, if the local probabilities $p_i$ were uniform,
$p_i= \left(\xymatrix@1@C=1em{ 1/2&1/2}\right)$, we would have:
\begin{align*}
  t_{a_0}&=\frac12&t_{a_1}&=\frac14&t_{a_2}&=\frac14&t_{a_3}&=\frac14&t_{a_4}&=\frac12.
\end{align*}
This choice, which may seem natural at first, would actually produce
a random heap with a bias, namely it would favor the appearance of
pieces $a_0$ and~$a_4$. Since $t_{a_0}=1/2$ and $t_{a_1}=1/4$ in this
case, it would produce on average twice more occurrences of $a_0$ than
occurrences of~$a_1$.

However, for uniform generation purposes, it is desirable to obtain
all coefficients $t_{a_i}$ \emph{equal}. Can we tune the initial
probability distributions $p_i$ in order to achieve this result?
Introduce the M\"obius polynomial $\mu(z)=1-5z+6z^2-z^3$, and consider
its root of smallest modulus, namely $q_0\approx0.308$. It turns out
that the only way to have all coefficients $t_{a_i}$ equal is to make
them precisely equal to~$q_0$. In turn, this imposes conditions on the
local probability distributions $p_1,p_2,p_3,p_4$ with only one
solution, yielding for this example:
\begin{align*}
\xymatrix@R=1em@C=7em{%
\strut  p_1(a_0)=q_0\approx 0.308\ar@{-}[r]^(.58){\cdot+\cdot=1}
&p_1(a_1)=1-q_0\approx0.692\strut\\
\strut p_2(a_1)=\frac{q_0}{1-q_0}\approx0.445
\ar@{-}[r]^(.58){\cdot+\cdot=1}%
\POS!R!U(.5)\ar@{-}[ur]!L!D(.5)^(.4){\cdot\times\cdot=q_0}
&p_2(a_2)=\frac{1-2q_0}{1-q_0}\approx0.555\strut
\\
\strut p_3(a_2)=\frac{1-2q_0}{1-q_0}\approx0.555
\ar@{-}[r]^(.58){\cdot+\cdot=1}%
\POS!R!U(.5)\ar@{-}[ur]!L!D(.5)^(.4){\cdot\times\cdot=q_0}
&p_3(a_3)=\frac{q_0}{1-q_0}\approx0.445\strut
\\
\strut p_4(a_3)=1-q_0\approx0.692
\ar@{-}[r]^(.58){\cdot+\cdot=1}%
\POS!U(.5)!R\ar@{-}[ur]!L!D(.5)^(.4){\cdot\times\cdot=q_0}
&\strut p_4(a_4)=q_0\approx0.308
}
\end{align*}

This array of positive numbers has the sought property that the
product of any two numbers along the depicted diagonals equals~$q_0$,
and the sum of every line equals~$1$; and $q_0$ is the only real
allowing this property for an array of this size. Running the \BSA\
with these values for the local probability distributions produces a
growing random heap which is uniform, in a precise meaning that will
be formalized later in the paper.

\subsection{Illustrating the Probabilistic Synchronization Algorithm
  (2)}
\label{sec:illustrating-bsa-2}

Consider the following network of alphabets with the four letters $a_0,a_1,a_2,a_3$:
\begin{align*}
  \Sigma_1&=\{a_0,a_1\},&\Sigma_2&=\{a_1,a_2\},&\Sigma_3&=\{a_2,a_3\},&\Sigma_4&=\{a_3,a_0\}.
\end{align*}

\begin{figure}
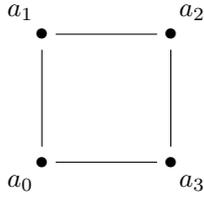

  \begin{align*}
\xy{%
{\xypolygon4"A"{~:{(2,2):(3,3):}{\hbox{\strut$\;\bullet\;$}}}%
\xypolygon4"B"{~:{(2,2):(4,4):}~><{@{}}{}}}
,"B2",*{a_0},
,"B3",*{a_3},
,"B4",*{a_2},
,"B1",*{a_1}
}%
\endxy
\end{align*}
  \caption{Synchronization graph of the ring model with four generators}
  \label{fig:qwqpijqwq}
\end{figure}

Analogously to the previous example, the synchronization graph for
this network is depicted in Figure~\ref{fig:qwqpijqwq}. Let us try to
apply the same generation technique as in the previous example. For
each alphabet~$\Sigma_i$, we consider a random sequence $Y_i$ of
letters of this alphabet. Given the symmetry of the network of
alphabets, in order to obtain a uniform distribution we need to
consider this time uniform local distributions
$p_i=\left(\xymatrix@1@C=1em{1/2&1/2}\right)$ for $i=1,2,3,4$. This
yields for instance:
\begin{align*}
  Y_1&=a_0a_1a_1a_0\dots&
  Y_2&=a_1a_2a_2a_1\dots&
  Y_3&=a_3a_2a_3a_2\dots&
Y_4&=a_0a_3a_0a_3\dots
\end{align*}
The construction of the stacking vector $Y$ with the matching of
letters yields:
\begin{gather*}
  Y=
\left(
  \begin{gathered}
\xymatrix@R=.5em@C=1ex{
a_0\ar@{-}[]!U+<0ex,1ex> &a_1\ar@{-}[d]&\dots\\
\ &a_1&a_2\ar@{-}[d]&\dots\\
&a_3\ar@{-}[d]&a_2&\dots\\
a_0\ar@{-}[]!D-<0ex,1ex>&a_3&\dots
}
  \end{gathered}
\right)\qquad\text{the two occurrences of $a_0$ being connected.}
\end{gather*}

\begin{figure}
\centering  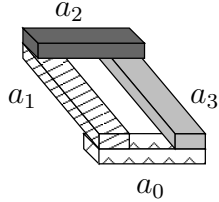
  \caption{Output heap resulting from an instance of the \BSA\ for the ring
    model with four generators}
  \label{fig:poqqqhja}
\end{figure}

The corresponding heap is depicted on Figure~\ref{fig:poqqqhja}. It is
then impossible to extend this heap while still respecting the
sequences $Y_1,Y_2,Y_3,Y_4$; not because some letter is waiting for
its matching, but because the different letters already present
are blocking each other, creating a cycle that prevents to interpret
the rest of the vector as an extension of the heap of
Figure~\ref{fig:poqqqhja}. Such a cycle will appear with probability
$1$ at some point, whatever the sequences
$Y_1,Y_2,Y_3,Y_4$. Henceforth the \BSA\ only outputs finite heaps for
this model.

\subsection{Illustrating the Probabilistic Full Synchronization Algorithm}
\label{sec:illustrating-fsa}

In the previous example, we have seen that the \BSA\ outputs a finite
heap with probability~$1$, whereas we would like to obtain heaps
arbitrary large. The Probabilistic Full Synchronization Algorithm
(\FSA) is designed to solve this issue. We will now describe how to
tune it, for the ring model with four generators introduced above and
illustrated in Figure~\ref{fig:qwqpijqwq}, in order to produce
infinite heaps uniformly distributed.

First, the theory of Bernoulli measures for trace monoids (see below
in Section~\ref{sec:prel-prob}) tells us that each piece of the monoid
must be given the probabilistic weight~$q_1$, root of smallest
modulus of the M\"obius polynomial $\mu(z)=1-4z+2z^2$, hence
$q_1=1-\sqrt2/2$.

Second, we arbitrarily choose a piece, say~$a_3$, to be removed. We are
left with the following network of alphabets: $\Sigma'_1=\{a_0,a_1\}$,
$\Sigma'_2=\{a_1,a_2\}$, $\Sigma'_3=\{a_2\}$,
$\Sigma'_4=\{a_0\}$. From the heaps point of view, this is equivalent
to the network with only two alphabets $(\Sigma'_1,\Sigma'_2)$, and
sharing thus the only letter~$a_1$. We will now design a variant of
the \BSA\ algorithm, to be applied to the network
$(\Sigma'_1,\Sigma'_2)$, outputting \emph{finite} heaps with
probability~$1$, and attributing the probabilistic parameter $q_1$ to
each of the three pieces $a_0,a_1,a_2$. This is possible by
considering \emph{sub-probability} distributions $p'_1,p'_2$
associated to $\Sigma'_1,\Sigma'_2$ instead of \emph{probability}
distributions for the local generations. A possible choice, although
not the unique one, is the following:
\begin{align*}
\xymatrix@R=1em@C=7em{%
\strut  p'_1(a_0)=q_1=1-\sqrt2/2\ar@{-}[r]^(.58){\cdot+\cdot=1}
&p'_1(a_1)=1-q_1=\sqrt2/2\strut\\
\strut p'_2(a_1)=\frac{q_1}{1-q_1}=\sqrt2-1
\ar@{-}[r]^(.58){\cdot+\cdot\,<1}%
\POS!R!U(.5)\ar@{-}[ur]!L!D(.5)^(.4){\cdot\times\cdot=q_1}
&p'_2(a_2)=q_1=1-\sqrt2/2\strut
}
\end{align*}
Note that the first line sums up to~$1$ whereas the second line sums
up to less than~$1$, and this guaranties that the \BSA\ executed with
these values will output a finite heap and stop with probability~$1$.

Let $\xi'_1$ be the output of the \BSA\ on $(\Sigma'_1,\Sigma'_2)$
with the above parameters. It is a finite heap built with occurrences
of $a_0$, $a_1$ and $a_2$ only. Then, we add the piece $a_3$
on top of~$\xi'_1$, and we check whether the obtained heap is
\emph{pyramidal}, which means that no piece can be removed from it
without moving the last piece~$a_3$ (see
Figure~\ref{fig:fasoaisjdq}). If it is not pyramidal, we reject it,
and re-run the \BSA, producing another instance of~$\xi'_1$, until
$\xi'_1\cdot a_3$ is pyramidal. At the end of this process, we obtain
a pyramidal heap $\xi_1=\xi'_1\cdot a_3$ built with the four
generators $a_0,a_1,a_2,a_3$, and a unique occurrence of~$a_3$.

\begin{figure}
\centering
\begin{tabular}{cc|cc}
  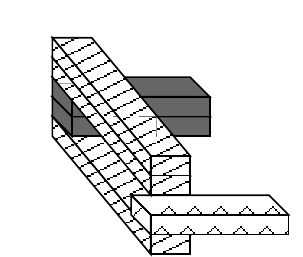\quad&
                        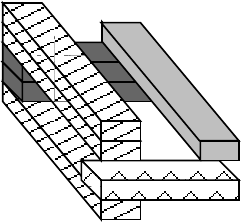\quad\strut&
                   \quad 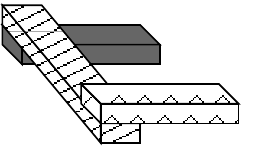\quad&
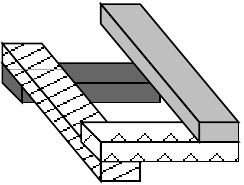
\\
(i) $\xi'_1$&(ii) reject&(iii) $\xi'_1$&(iv) accept
\end{tabular}
\caption{First iteration of the \FSA. (i) The finite output $\xi'_1$
  of the \BSA\ on three generators $a_0,a_1,a_2$. (ii)~Adding the
  piece~$a_3$ on top of~$\xi'_1$, the obtained heap is not pyramidal
  since occurrences of $a_1$ can be removed without
  moving~$a_3$. (iii)~Another instance of the heap $\xi'_1$ on three
  generators $a_0,a_1,a_2$. (iv)~This time, the new heap $\xi_1$
  obtained by adding $a_3$ on top of $\xi'_1$ is pyramidal.}
  \label{fig:fasoaisjdq}
\end{figure}

The process of computing this heap $\xi_1$ is the first loop of the
\FSA. We initialize what will be the final output of the algorithm by
setting $y_1=\xi_1$.  We then reproduce the above random procedure,
yielding a pyramidal heap~$\xi_2$ and then we form the heap
$y_2=y_1\cdot\xi_2$, obtained by simply piling up $\xi_2$ on top
of~$y_1$. We iterate this procedure: outputting~$\xi_3$ pyramidal, we
put $y_3=y_2\cdot\xi_3$, and so on. Then the increasing heaps
$y_1,y_2,\ldots$ approximate an infinite heap on the four generators
$a_0,a_1,a_2,a_3$ which we claim to be uniformly distributed.

\section{Preliminaries on trace monoids}
\label{sec:prel-trace-mono}

An \emph{alphabet} is a finite set, usually denoted by~$\Sigma$, the
elements of which are called \emph{letters}. The free monoid generated
by $\Sigma$ is denoted by~$\Sigma^*$\,.

\subsection{Combinatorics of trace monoids}
\label{sec:trace-monoids}

\subsubsection{Definitions}
\label{sec:definitions}

An \emph{independence pair} on the alphabet $\Sigma$ is a binary,
irreflexive and symmetric relation on~$\Sigma$, denoted~$I$. The
\emph{trace monoid} $\M=\M(\Sigma,I)$ is the presented monoid
$\M=\langle\Sigma\;|\; ab=ba\text{ for all $(a,b)\in
  I$}\rangle$.
Hence, if $\R$ is the smallest congruence on $\Sigma^*$ that contains
all pairs $(ab,ba)$ for $(a,b)$ ranging over~$I$, then $\M$ is the
quotient monoid $\M=\Sigma^*/\R$. Elements of $\M$ are called
\emph{traces}~\cite{diekert90}. The unit element is denoted
by~$\unit$, and the concatenation in $\M$ is denoted by~``$\cdot$''.

The \emph{length} $|x|$ of a trace $x\in\M$ is the length of any word
in the equivalence class~$x$. The left divisibility relation on~$\M$
is denoted by~$\leq$, it is defined by $x\leq y\iff\exists z\in\M\quad
y=x\cdot z$.

\subsubsection{Cliques}
\label{sec:cliques}

A \emph{clique} of $\M$ is any element $x\in\M$ of the form
$x=a_1\cdot\ldots\cdot a_n$ with $a_i\in\Sigma$ and such that
$(a_i,a_j)\in I$ for all $i\neq j$. Cliques thus defined are in
bijection with the cliques, in the graph-theoretic sense, of the pair
$(\Sigma,I)$ seen as an undirected graph, that is to say, with the set of
complete sub-graphs of~$(\Sigma,I)$.

The set of cliques is denoted by~$\C$. The unit element is a clique,
called the empty clique. The set of non empty cliques is denoted
by~$\Cstar$. 

For instance, for
$\M=\langle a_0,a_1,a_2,a_3,a_4\;|\;a_ia_j=a_ja_i\ \text{for }
|i-j|>1\rangle$,
we have
$\C=\{\unit,\;a_0,\;a_1,\;a_2,\;a_3,\;a_4,\;a_0a_2,\;a_0a_3,\;a_0a_4,\;a_1a_3,\;a_1a_4,\;a_2a_4,\;
a_0a_2a_4\}$.
We call this trace monoid the \emph{path model with five generators},
it corresponds to the example introduced in
Section~\ref{sec:bsa-algorithm}. Note that the synchronization graph
depicted in Figure~\ref{fig:poqwpqq} is not $(\Sigma,I)$ but its
complementary.

\subsubsection{Growth series and M\"obius polynomial}
\label{sec:growth-series-mobius}

The \emph{growth series} of $\M$ is the formal series
\begin{gather*}
  Z_\M(t)=\sum_{x\in\M}t^{|x|}\,.
\end{gather*}
The \emph{M\"obius polynomial}~\cite{cartier69} is the polynomial $\mu_\M(t)$ defined
by
\begin{gather*}
  \mu_\M(t)=\sum_{c\in\C}(-1)^{|c|}t^{|c|}\,.
\end{gather*}

For the path model with five generators introduced above, we have $\mu_\M(t)=1-5t+6t^2-t^3$.

The M\"obius polynomial is the formal inverse of the growth series~\cite{cartier69,viennot86}:
\begin{gather*}
  Z_\M(t)=1/\mu_\M(t)\,.
\end{gather*}

The M\"obius polynomial has a unique root of smallest modulus,
say~$p_0$\,. This root is real and lies in $(0,1]$, and coincides with
the radius of convergence of the power
series~$Z_\M(t)$~\cite{krob03,goldwurm00:_clique}.

\subsubsection{Multivariate M\"obius polynomial}
\label{sec:multi-variate-mobius}

The M\"obius polynomial has a multivariate version,
$\mu_\M(t_1,\ldots,t_N)$ where $t_1,\ldots,t_N$ are formal variables
associated with the generators $a_1,\ldots,a_N$ of~$\M$. It is defined by:
\begin{gather*}
  \mu_\M(t_1,\ldots,t_N)=\sum_{c\in\C}(-1)^{|c|} t_{j_1}\cdot\ldots\cdot t_{j_c}\,,
\end{gather*}
where the variables $t_{j_1},\ldots,t_{j_c}$ correspond to the letters
such that $c=a_{j_1}\cdot\ldots\cdot a_{j_c}$\,.

\subsubsection{Irreducibility}
\label{sec:irreducibility}

The \emph{dependence relation} associated with $\M=\M(\Sigma,I)$ is
the binary, symmetric and reflexive relation
$D=(\Sigma\times\Sigma)\setminus I$. The trace monoid $\M$ is
\emph{irreducible} if the pair $(\Sigma,D)$ is connected as a
non oriented graph.

If $\M$ is irreducible, then the root $p_0$ of $\M$ is simple~\cite{krob03}.

\subsection{The heap of pieces interpretation of traces}
\label{sec:heap-piec-interpr}

\subsubsection{The picture}
\label{sec:picture}

Viennot's theory provides a visualization of traces as \emph{heaps of
  pieces}~\cite{viennot86}.  Picture each trace as the piling of
dominoes labeled by the letters of the alphabet, and such that
dominoes associated with two letters $a$ and $b$ fall to the ground
according to parallel lanes, and disjoint if and only if $(a,b)\in I$.
See an illustration on Figure~\ref{fig:congruenyt}, (i)--(iii), for
$\M=\langle a,b,c,d\;|\; ac=ca,\ ad=da,\ bd=db\rangle$\,.

\begin{figure}
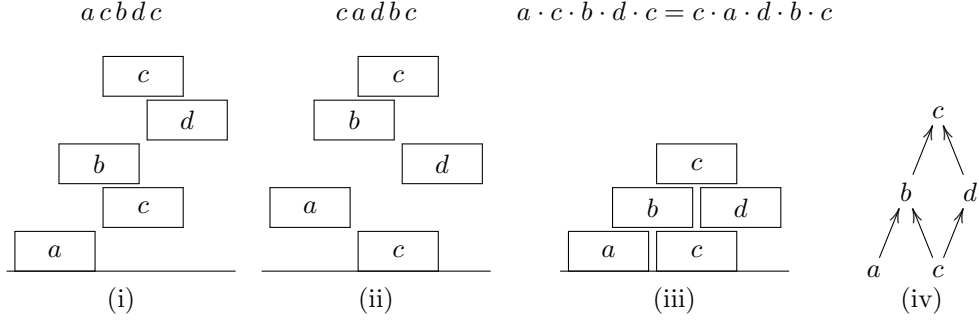

  \centering
\begin{gather*}
\begin{array}{cccc}
a\,c\,b\,d\,c&c\,a\,d\,b\,c&a\cdot c\cdot b\cdot d\cdot c=c\cdot
                             a\cdot d\cdot b\cdot c\\[1em]
\xy<.15em,0em>:
(-2,0);(55,0)**@{-},
0="G",
@={"G","G"+(20,0),"G"+(20,10),"G"+(0,10)},
s0="prev" @@{;"prev";**@{-}="prev"}, "G"+(10,5)*{a},
(22,11)="G",
@={"G","G"+(20,0),"G"+(20,10),"G"+(0,10)},
s0="prev" @@{;"prev";**@{-}="prev"}, "G"+(10,5)*{c},
(11,22)="G",
@={"G","G"+(20,0),"G"+(20,10),"G"+(0,10)},
s0="prev" @@{;"prev";**@{-}="prev"}, "G"+(10,5)*{b},
(33,33)="G",
@={"G","G"+(20,0),"G"+(20,10),"G"+(0,10)},
s0="prev" @@{;"prev";**@{-}="prev"}, "G"+(10,5)*{d},
(22,44)="G",
@={"G","G"+(20,0),"G"+(20,10),"G"+(0,10)},
s0="prev" @@{;"prev";**@{-}="prev"}, "G"+(10,5)*{c},
\endxy
&
\xy<.15em,0em>:
(-2,0);(55,0)**@{-},
(0,11)="G",
@={"G","G"+(20,0),"G"+(20,10),"G"+(0,10)},
s0="prev" @@{;"prev";**@{-}="prev"}, "G"+(10,5)*{a},
(22,0)="G",
@={"G","G"+(20,0),"G"+(20,10),"G"+(0,10)},
s0="prev" @@{;"prev";**@{-}="prev"}, "G"+(10,5)*{c},
(11,33)="G",
@={"G","G"+(20,0),"G"+(20,10),"G"+(0,10)},
s0="prev" @@{;"prev";**@{-}="prev"}, "G"+(10,5)*{b},
(33,22)="G",
@={"G","G"+(20,0),"G"+(20,10),"G"+(0,10)},
s0="prev" @@{;"prev";**@{-}="prev"}, "G"+(10,5)*{d},
(22,44)="G",
@={"G","G"+(20,0),"G"+(20,10),"G"+(0,10)},
s0="prev" @@{;"prev";**@{-}="prev"}, "G"+(10,5)*{c},
\endxy
&
\xy<.15em,0em>:
(-2,0);(55,0)**@{-},
0="G",
@={"G","G"+(20,0),"G"+(20,10),"G"+(0,10)},
s0="prev" @@{;"prev";**@{-}="prev"}, "G"+(10,5)*{a},
(22,0)="G",
@={"G","G"+(20,0),"G"+(20,10),"G"+(0,10)},
s0="prev" @@{;"prev";**@{-}="prev"}, "G"+(10,5)*{c},
(11,11)="G",
@={"G","G"+(20,0),"G"+(20,10),"G"+(0,10)},
s0="prev" @@{;"prev";**@{-}="prev"}, "G"+(10,5)*{b},
(33,11)="G",
@={"G","G"+(20,0),"G"+(20,10),"G"+(0,10)},
s0="prev" @@{;"prev";**@{-}="prev"}, "G"+(10,5)*{d},
(22,22)="G",
@={"G","G"+(20,0),"G"+(20,10),"G"+(0,10)},
s0="prev" @@{;"prev";**@{-}="prev"}, "G"+(10,5)*{c},
\endxy
&
\xy<.15em,0em>:
0*+{a};(8,20)*+{b}**@{-}?>*\dir{>},
(16,0)*+{c};(24,20)*+{d}**@{-}?>*\dir{>},(8,20)*+{\phantom{b}}**@{-}?>*\dir{>},
(8,20)*+{\phantom{b}};(16,40)*+{c}**@{-}?>*\dir{>},
(24,20)*+{\phantom{d}};(16,40)*+{\phantom{c}}**@{-}?>*\dir{>},
\endxy
\\
\text{(i)}&\text{(ii)}&\text{(iii)}&\text{(iv)}
\end{array}
\end{gather*}
\caption{{\normalfont(i)--(ii)}: representation of the two
    congruent words $acbdc$ and $cadbc$ in $\langle a,b,c,d\; |\;
    ac=ca,\ ad=da,\ bd=db\rangle$.\quad  {\normalfont(iii)}: the resulting
    trace $a\cdot c\cdot b\cdot d\cdot c$, represented as a heap of
    pieces.\quad {\normalfont(iv)}: the associated labeled ordered
    poset. For instance $a<d$ does not hold since $acbdc\to cabdc\to
    cadbc\to cdabc$.}
\label{fig:congruenyt}
\end{figure}

\subsubsection{Traces as labeled ordered sets}
\label{sec:traces-as-labelled}

The formalization of this picture is done by interpreting each trace
$x$ of a trace monoid as an equivalence class, up to isomorphism, of a
$\Sigma$-labeled partial order, which we describe now.

Let $x$ be a trace, written as a product of letters
$x=a_1\cdot\ldots\cdot a_n$\,. Let $\xbar=\{1,\ldots,n\}$, and define
a labeling $\phi:\xbar\to\Sigma$ by $\phi(i)=a_i$\,. Equip $\xbar$
with the natural ordering on integers~$\leq$. Then remove the pair
$(i,j)$ from the strict ordering relation $<$ on $\xbar$ whenever, by
a finite number of adjacent commutations of distinct letters, the
sequence $(a_1,\ldots,a_n)$ can be transformed into a sequence
$(a_{\sigma(1)},\ldots,a_{\sigma(n)})$ such that
$\sigma(i)>\sigma(j)$. 

The set of remaining ordered pairs $(i,j)$ is a partial ordering
on~$\xbar$, which only depends on~$x$, and which defines the heap
associated with~$x$. See Figure~\ref{fig:congruenyt},~(iv).

\subsection{Completing trace monoids with infinite traces}
\label{sec:infinite-traces-1}

\subsubsection{Infinite traces}
\label{sec:infinite-traces}

Let $\seq xn$ and $\seq yn$ be two nondecreasing sequences of traces: $x_n\leq
x_{n+1}$ and $y_n\leq y_{n+1}$ for all integers $n\geq1$. We identify
$\seq xn$ and $\seq yn$ whenever they satisfy:
\begin{align*}
(  \forall n\geq1\quad\exists m\geq1\quad x_n\leq y_m)\wedge
(\forall n\geq1\quad\exists m\geq1\quad y_n\leq x_m)\,.
\end{align*}

This identification is an equivalence relation between nondecreasing
sequences. The quotient set is denoted by~$\Mbar$. It is naturally
equipped with a partial ordering, such that the mapping $\M\to\Mbar$
associating the (equivalence class of the) constant sequence $x_n=x$
to any trace $x\in\M$, is an embedding of partial orders.

We identify thus $\M$ with its image in $\Mbar$ through the embedding
$\M\to\Mbar$. The set $\BM$ of \emph{infinite traces} is defined by:
\begin{gather*}
  \BM=\Mbar\setminus\M\,.
\end{gather*}

The set $\BM$ is called the \emph{boundary} of~$\M$~\cite{abbes_mair15}. Visually,
elements of $\BM$ correspond to \emph{infinite countable heaps}, that
is to say, limiting heaps obtained by piling up countable many
pieces.

\subsubsection{Properties of $(\Mbar,\leq)$}
\label{sec:properties-mbar-leq}

The partially ordered set $(\Mbar,\leq)$ is complete with respect to
least upper bound of nondecreasing sequences. And any element
$\xi\in\Mbar$ is the least upper bound of a nondecreasing sequence of
elements of~$\M$~\cite{abbes_mair15}.


\subsection{Synchronization of sequences}
\label{sec:prod-free-autom}

In this section we introduce an alternative way of looking at trace
monoids, by means of vectors of words. This emphasizes the distributed
point of view on trace monoids.

\subsubsection{Network of alphabets and synchronizing trace
  monoid.}
\label{sec:synchr-alph-synchr}

A family $(\Sigma_1,\ldots,\Sigma_N)$ of $N$ alphabets, not
necessarily disjoint, is called a \emph{network of alphabets}. Let
$\Sigma=\Sigma_1\cup\ldots\cup\Sigma_N$\,. Define $R=\{1,\ldots,N\}$
as the set of \emph{resources}, and consider the product monoid
\begin{gather}
\label{eq:36}
  \H=(\Sigma_1)^*\times\ldots\times(\Sigma_N)^*\,.
\end{gather}

Identify each $a\in\Sigma$ with the element $H(a)=(a_i)_{i\in R}$ of $\H$
defined by:
\begin{gather*}
  \forall i\in R\quad a_i=
  \begin{cases}
    \epsilon\text{ (the empty word)},&\text{if $a\notin\Sigma_i$\,,}\\
    a,&\text{if $a\in\Sigma_i$\,.}
  \end{cases}
\end{gather*}

For each letter $a\in\Sigma$, the \emph{set of resources} associated
with $a$ is the subset $R(a)$ defined by
\begin{gather*}
  R(a)=\{i\in R\tq a\in\Sigma_i\}\,.
\end{gather*}

Let $\G$ be the sub-monoid of $\H$ generated by the collection $\{H(a)\tq
a\in\Sigma\}$. Then $\G$ is isomorphic with the trace monoid
$\M=\M(\Sigma,I)$, where the independence relation $I$ on $\Sigma$ is
defined by
\begin{gather}
  \label{eq:37}
  (a,b)\in I\iff R(a)\cap R(b)=\emptyset.
\end{gather}

The monoid $\M$ is called the \emph{synchronization trace monoid}, or
simply the \emph{synchronization monoid}, of the network
$(\Sigma_1,\ldots,\Sigma_N)$. Examples have been given in
Section~\ref{sec:illustrating-bsa-fsa}. It can be proved that every
trace monoid is isomorphic to a synchronization trace
monoid~\cite{cori85}.

\subsubsection{Synchronization of sequences.}
\label{sec:synchr-trace-assoc}

Let $(\Sigma_1,\dots,\Sigma_N)$ be a network of alphabets. Each monoid
$(\Sigma_i)^*$ being equipped with the prefix ordering, we equip the
product monoid $\H$ defined in~(\ref{eq:36}) with the product order,
and the sub-monoid $\G$ with the induced order.

Assume given a vector of sequences $Y=(Y_1,\dots,Y_N)\in\H$. The
\emph{synchronization trace} of $Y$ is defined as the following least
upper bound in~$\G$:
\begin{gather}
\label{eq:39}
  X=\bigvee\{Z\in\G\tq Z\leq Y\}.
\end{gather}
This least upper bound does indeed exist in~$\G$: it is obtained by
taking least upper bounds component-wise, which are all well defined.

Identifying $\G$ with the trace monoid $\M$ defined in~(\ref{eq:37})
and above, the element $X$ is thus the largest \emph{trace} among
those tuples below~$Y$.

For example, consider the ring model with four generators introduced
in Section~\ref{sec:illustrating-bsa-2}, with $\Sigma_1=\{a_0,a_1\}$,
$\Sigma_2=\{a_1,a_2\}$, $\Sigma_3=\{a_2,a_3\}$ and
$\Sigma_4=\{a_3,a_0\}$, and the vector:
\begin{gather*}
  Y=
  \begin{pmatrix}
a_0a_1a_1a_0\\
a_1a_2a_2a_1\\
a_3a_2a_3a_2\\
a_0a_3a_0a_3    
  \end{pmatrix}
\end{gather*}

One easily convinces oneself that the synchronization of $Y$ is the following:
\begin{gather*}
  X=
  \begin{pmatrix}
a_0a_1\\
a_1a_2\\
a_3a_2\\
a_0a_3    
  \end{pmatrix}\quad\text{corresponding to $a_0\cdot a_1\cdot a_3\cdot
    a_2$ in the trace monoid.}
\end{gather*}
We shall see below an algorithmic way to determine the synchronization
of any given vector $Y\in\H$.

\subsubsection{On-line computation features.}
\label{sec:algorithmic-issues}

The algorithmic computation of the synchronization trace of a given
vector of sequences might belong to the folklore of concurrency
theory, see for instance~\cite{cori85}. In what follows, an special
emphasis is given to the distributed and on-line nature of this
computation.

The remaining of this section is devoted to provide an algorithm
taking a vector of sequences $Y\in\H$ as input, and outputting the
synchronization trace of~$Y$. Furthermore, the input trace $Y$ might
not be entirely known at the time of computation. Instead, we assume
that only a sub-trace $Y'\leq Y$ feeds the algorithm, together with
the knowledge whether some further input is about to come or not.
Therefore, we need our algorithm to produce the \emph{two} following
outputs:
\begin{enumerate}
\item The best approximation $X'$ of the synchronization trace~$X$,
  given only the input~$Y'$, which is merely the synchronization trace
  of~$Y'$.
\item And one of the following tags:
  \begin{enumerate}
  \item[\DL] for ``Deadlock'' if some more input is about to come, and
    yet the algorithm has already reached the synchronization trace
    of~$Y$, \emph{i.e.}, if $X'=X$ whatever the continuation of~$Y'$.
  \item[\WFI] for ``Waiting for Input'' if some more input is about to
    come, and the trace $X'$ might be extended.
  \item[\EOF] for ``End of file'' if there is no more input to come.
  \end{enumerate}
\end{enumerate}

\subsubsection{Detecting deadlocks}
\label{sec:detecting-deadlocks}

Consider, as in
Sections~\ref{sec:synchr-alph-synchr}---\ref{sec:synchr-trace-assoc},
a network $(\Sigma_1,\dots,\Sigma_N)$ of alphabets. Given a finite
word $Y_i\in(\Sigma_i)^*$, we denote by $\Ybar_i$ the word
$\Ybar_i=Y_i\cdot\dag_i$, where $\dag_i$ is an additional symbol which
can be either $\EOF$ of~$\WFI$, to be interpreted as ``no more input
will ever come'' if $\dag_i=\EOF$, or as ``some more input might
arrive'' if $\dag_i=\WFI$. Finally, let
$\Ybar=(\Ybar_1,\dots,\Ybar_N)$.

We first introduce a basic routine described in pseudo-code below
(Algorithm~\ref{algo1}). The input of the routine is the
vector~$\Ybar$. Its output is as follows:
\begin{enumerate}
\item\label{item:13} If the synchronization trace $X$ of $Y$ is non empty, then the
  routine outputs a minimal piece of~$X$.
\item\label{item:14} If the synchronization trace $X$ of $Y$ is empty, then the
  routine outputs a flag $\ddag$ explaining why $X$ is empty:
\begin{enumerate}
\item\label{item:15} If $\Ybar=(\EOF,\dots,\EOF)$, then $\ddag=\EOF$. Note that this
  necessarily entails that $Y$ is empty.
\item\label{item:16} In case at least one input component of $\Ybar$
  does not carry the symbol $\dag_i=\EOF$, then:
  \begin{enumerate}
  \item\label{item:17} If adding some letters to some of the components of $Y$
    carrying $\dag_i=\WFI$ could yield a non empty synchronization
    trace, then $\ddag=\WFI$.
  \item\label{item:18} If no letter can be added to the components of $Y$ carrying
    $\dag_i=\WFI$ to yield a non empty synchronization trace, then
    $\ddag=\DL$.
  \end{enumerate}
\end{enumerate}
\end{enumerate}

To detect if a piece $u$ is minimal in a heap, when the heap is given
by its representation in~$\G$, it is necessary and sufficient to check
that it is minimal in all the coordinates where it occurs, that is to
say, in all the components belonging to the set of
resources~$\R(u)$. This justifies that Algorithm~\ref{algo1} is
sound. Observe that the algorithm is non deterministic, as there may
be several minimal pieces.

\begin{algorithm}
\caption{Determines a minimal piece of the synchronizing trace $X$ of $Y=(Y_1,\dots,Y_N)$}
\label{algo1}
\begin{algorithmic}[1]
\Require{$\Ybar_1,\dots,\Ybar_N$}
\Comment{Recall that $\Ybar_i=Y_i\cdot\dag_i$}
\ForAll{$i\in\{1,\dots,N\}$}
\State
$u_i\gets{\text{first letter of $\Ybar_i$}}$\Comment{$u_i$ is either a
real letter or $\dag_i$}
\EndFor
\State $H\gets\{i\tq u_i\neq\EOF\}$\Comment{The set of indices of
  interest}
\State $\Hbar\gets\{i\tq u_i=\EOF\}$\Comment{The complementary set of $H$}
\If{$H=\emptyset$}\Comment{Case~\ref{item:15} of the above discussion}
\State\Return{\EOF}
\EndIf
\State $K\gets\{i\in H\tq u_i\neq\WFI\}$\Comment{The set of indices
  with real letters} 
\If{$K=\emptyset$}\Comment{One instance of Case~\ref{item:17}}
\State\Return{\WFI}
\EndIf
\ForAll{$i\in K$}
\If{$\R(u_i)\cap\Hbar\neq\emptyset$}\Comment{No chance to obtain later
  the\\
  \hfill expected synchronization for $u_i$} 
\State $M_i\gets\DL$
\Else
\If{$u_j=u_i$ for all $j\in\R(u_i)$}\Comment{Case where $u_i$ is minimal\\
  \hfill in all the expected components of $Y$}
\State $M_i\gets u_i$
\Else 
\If{$u_j\in\{u_i,\WFI\}$ for all
  $j\in\R(u_i)$}\Comment{Synchronization\\
  \hfill for $u_i$ is possible in the future}
\State $M_i\gets\WFI$
\Else
\State
$M_i\gets\DL$\Comment{The piece $u_i$ is not minimal in $X$}
\EndIf
\EndIf
\EndIf
\EndFor
\If{$M_i=\DL$ for all $i\in K$}\Comment{Case~\ref{item:18}:  the synchronization\\
  \hfill trace $X$ is empty since it has no minimal piece} 
\State\Return{\DL}
\Else
\If{$M_i=\WFI$ for all $i\in K$}\Comment{Case~\ref{item:17} (again)}
\State
\Return $\WFI$
\Else\Comment{Case~\ref{item:13}}
\State\Return one $M_i$ with
$M_i\notin\{\DL,\WFI\}$
\Comment{Any
  $M_i\notin\{\DL,\WFI\}$ is minimal\\
  \hfill in the synchronization trace $X$}
\EndIf
\EndIf
\end{algorithmic}
\end{algorithm}

Algorithm~\ref{algo1} executes in constant time, the constant growing
linearly with the number of edges in the dependence graph $(\Sigma,D)$
of the associated trace monoid. If communicating processes are
devised, one per each alphabet, and able to write to a common
register, they will be able to perform Algorithm~\ref{algo1} in a
distributed way. We shall not work out the details of the distributed
implementation, since it would be both outside the scope of the paper
and out of the range of expertise of the author. People from the
distributed algorithms community will probably find it routine.

\subsubsection{Computation of the synchronization trace.}
\label{sec:comp-synchr-trac}

With Algorithm~\ref{algo1} at hand, we can now give an algorithm to
compute the synchronization trace of a given vector of sequences
$Y\in\H$. This is the topic of the Synchronization Algorithm
(Algorithm~\ref{algo2}), of which we give the pseudo-code below, and
which takes as input a vector $\Ybar=(\Ybar_1,\dots,\Ybar_N)$ of the
same kind as Algorithm~\ref{algo1}.

The Synchronization Algorithm iteratively executes
Algorithm~\ref{algo1}, and collects the minimal pieces thus obtained
to form it own output. If $Y=(Y_1,\dots,Y_N)\in\H$ and if $M\leq Y$,
then we denote by $Z=M\backslash Y$ the unique vector $Z\in\H$ such that
$M\cdot Z =Y$.

As in the requirements stated above in
Section~\ref{sec:algorithmic-issues}, the Synchronization Algorithm
outputs both the synchronization trace of its input and a tag
advertising if the computation is over, either because a deadlock has
been reached or because the input feed is over, or if it is still
waiting for some input that might cause the synchronization trace
computed so far to be extended. This feature will be crucial when
considering its execution on sequences of letters which are possibly
infinite, see below in Section~\ref{sec:feed-synchr-algor}.

The Synchronization Algorithm can be executed in a distributed way by
$N$ communicating processes, one for each coordinate, and it runs in
time linear with the size of the synchronization trace~$X$.

\begin{algorithm}
\caption{Synchronization Algorithm: computes the synchronization trace
of $Y$}
\label{algo2}
\begin{algorithmic}[1]
\Require $\Ybar=(\Ybar_1,\dots,\Ybar_N)$\Comment{$\Ybar_i=Y_i\cdot\dag_i$}
\State $X\gets\unit$\Comment{Initialize the variable $X$ with the empty
  heap}
\State \textbf{call} Algorithm~\ref{algo1} with input $\Ybar$
\State $M\gets\text{output of Algorithm~\ref{algo1}}$
\While{$M\notin\{\DL,\WFI,\EOF\}$}
\State $X\gets X\cdot M$
\State $\Ybar\gets M\backslash \Ybar$
\State \textbf{call} Algorithm~\ref{algo1} with input $\Ybar$
\State $M\gets\text{output of Algorithm~\ref{algo1}}$
\EndWhile
\State\Return $(X,M)$
\end{algorithmic}
\end{algorithm}

\subsubsection{Feeding the Synchronization Algorithm with a possibly
  infinite input}
\label{sec:feed-synchr-algor}

We have defined the synchronization trace of a vector $Y\in\H$ in
Section~\ref{sec:synchr-trace-assoc}. The same definition applies if
one or several components of $Y$ are infinite, by putting:
\begin{gather}
  \label{eq:38}
  X=\bigvee\{Z\in\G\tq Z\leq Y\},
\end{gather}
and this least upper bound is always well defined in $\Mbar$ since
least upper bounds of nondecreasing sequences always exist in
$(\Mbar,\leq)$ as recalled in
Section~\ref{sec:properties-mbar-leq}. We call the element $X$ of
$\Mbar$ thus defined the \emph{generalized synchronization trace}
of~$Y$.

Adapting the algorithmic point of view to an infinite input introduces
obviously some issues, which can be addressed by observing that the
output of the Synchronization Algorithm is nondecreasing with its
input. We will idealize the situation where the Synchronization
Algorithm is repeatedly fed with a nondecreasing input by
saying that it is fed with a vector with possibly infinite
components. The output, possibly infinite, is defined as the least
upper bound in $\Mbar$ of the nondecreasing sequence of finite output
heaps. It coincides of course with the generalized synchronization
trace $X$ defined in~(\ref{eq:38}).

To describe more precisely this algorithmic procedure, consider a
vector $Y$ of sequences, some of which may be infinite. We assume that
$Y$ is effectively given through some \emph{sampling}, \emph{i.e.}, as
an \emph{infinite} concatenation of \emph{finite} vectors $Z_k\in\H$:
\begin{gather*}
  Y=Z_1\cdot Z_2\cdot Z_3\cdots
\end{gather*}
If a component of $Y$ is finite, the corresponding component of the
vector $Z_k$ will be the empty word for $k$ large enough. 

Furthermore, we assume that a primitive is able to produce, for each
integer~$k\geq1$, a vector $\Zbar{}_k=Z_k\cdot\dag_k$, where $\dag_k$
is itself a vector $\dag_k=(\dag_{k,i})_i$ with
$\dag_{k,i}\in\{\EOF,\WFI\}$, in such a way that an occurrence of $\EOF$
marks the finiteness of the corresponding component of~$Y$. Formally,
we assume that the following two properties hold, for all components:
\begin{gather*}
  (\text{$Y_i$ is a finite sequence})\iff
(\exists k\geq1\quad
  \dag_{k,i}=\EOF)\\
\text{and}\quad
\forall k\geq 1\quad \dag_{k,i}=\EOF\implies(\forall k'\geq k\quad\dag_{k',i}=\EOF)
\end{gather*}

The Generalized Synchronization Algorithm, the pseudo-code of which is
given below in Algorithm~\ref{algo3}, is then recursively fed with
$\Zbar_1,\Zbar_2,\dots$, writing out to its output register~$X$. The
{output} of Algorithm~\ref{algo3} is the least upper bound,
in~$\Mbar$, of the sequence of heaps that recursively appear in the
register~$X$.

The Generalized Synchronization Algorithm exits its \textbf{while}
loop if and only if the generalized synchronization trace is
finite. In all cases, its output (as defined above) is the generalized
synchronization trace of the vector~$Y$, regardless of the
decomposition $(\Zbar_k)_{k\geq1}$ that feeds its input.

\begin{algorithm}
\caption{Generalized Synchronization Algorithm}
\label{algo3}
\begin{algorithmic}[1]
\Require $(\Zbar_k)_{k\geq1}$\Comment{The $Z_k$s recursively feed the
  input}
\State $X\gets\unit$\Comment{Initialize the variable $X$ with the
  empty heap}
\State $k\gets1$
\State \textbf{call} Algorithm~\ref{algo2} with input $\Zbar_1$
\State $(U,\dag)\gets\text{output of
  Algorithm~\ref{algo2}}$\Comment{$\dag\in\{\EOF,\DL,\WFI\}$}
\While{$\dag=\WFI$}
\State $X\gets X\cdot U$
\State $k\gets k+1$
\State \textbf{call} Algorithm~\ref{algo2} with input $\Zbar_k$
\State $(U,\dag)\gets\text{output of
  Algorithm~\ref{algo2}}$
\EndWhile
\end{algorithmic}
\end{algorithm}

\section{Preliminaries on probabilistic trace monoids}
\label{sec:prel-prob}

We denote by $\bbR_+^*$ the set of positive reals.

\subsection{Bernoulli and finite Bernoulli sequences}
\label{sec:bern-finite-bern}

In this section we collect classical material found in many
textbooks~\cite{billingsley95}. We pay a special attention to
presenting this material so as to prepare for its generalization to
trace monoids.

\subsubsection{Bernoulli sequences}
\label{sec:bernoulli-sequences}

Classically, a \emph{Bernoulli sequence} on an alphabet $\Sigma$ is an
infinite sequence $\seq Xn$ of independent and identically distributed
(\iid) random variables, where each $X_i$ takes its values
in~$\Sigma$. In order to eliminate degenerated cases, we assume that
the common probability distribution, say~$\rho$, over~$\Sigma$ of
all~$X_i$ is \emph{positive} on~$\Sigma$\,; hence $\rho$ is bound to
satisfy:
\begin{align*}
  \forall a\in\Sigma\quad \rho_a&>0\,,&\sum_{a\in\Sigma}\rho_a&=1\,.
\end{align*}

\subsubsection{Bernoulli measures}
\label{sec:stand-prob-space}

The canonical probability space associated with the Bernoulli sequence
$\seq Xn$ is the triple $(\B\Sigma^*,\FFF,\pr)$ defined as follows: the
set $\B\Sigma^*$ is the set of infinite sequences with values
in~$\Sigma$. The \slgb\ $\FFF$ is the \slgb\ generated by the
countable collection of \emph{elementary cylinders}~$\up x$, for $x$
ranging over the free monoid~$\Sigma^*$, and defined by
\begin{gather*}
  \up x=\{\omega\in\B\Sigma^*\tq x\leq\omega\}\,,
\end{gather*}
where $x\leq\omega$ means that the infinite sequence $\omega$ starts
with the finite word~$x$. Finally, $\pr$~is the unique probability
measure on $(\B\Sigma^*,\FFF)$ which takes the following values on
elementary cylinders:
\begin{gather*}
  \forall x\in\Sigma^*\quad\pr(\up x)=f(x)\,;
\end{gather*}
here $f:\Sigma^*\to\bbR_+^*$ is the unique positive function
satisfying:
\begin{align}
\label{eq:1}
  \forall a\in\Sigma\quad f(a)&=\rho_a\,,&\forall x,y\in\Sigma^*\quad f(xy)&=f(x)f(y)\,,
\end{align}
where $xy$ denotes the concatenation of words $x$ and~$y$. Bernoulli
sequences correspond exactly to probability measures $\pr$ on
$(\B\Sigma^*,\FFF)$ with the following property:
\begin{gather*}
\forall x\in\Sigma^*\quad\pr(\up x)>0\,,\\
  \forall x,y\in\Sigma^*\quad\pr\bigl(\up(xy)\bigr)=\pr(\up x)\pr(\up y)\,.
\end{gather*}
Such probability measures on $(\B\Sigma^*,\FFF)$ are called
\emph{Bernoulli measures}.

\subsubsection{Uniform Bernoulli measure}
\label{sec:unif-bern-meas}

Among Bernoulli measures associated with the alphabet~$\Sigma$, one
and only one is \emph{uniform}, in the following sense:
\begin{gather*}
  \forall x,y\in\Sigma^*\quad|x|=|y|\implies\pr(\up x)=\pr(\up y)\,,
\end{gather*}
where $|x|$ denotes the length of the word~$x$.  It is characterized
by $\pr(\up x)=p_0^{|x|}$\,, where $p_0=1/|\Sigma|$\,.

\subsubsection{Finite Bernoulli sequences}
\label{sec:finite-bern-sequ}

Let $\rho=(\rho_a)_{a\in\Sigma}$ be a sub-probability distribution
over~$\Sigma$, hence bound to satisfy:
\begin{align*}
  \forall a\in\Sigma\quad\rho_a&>0\,,&
\sum_{a\in\Sigma}\rho_a&<1\,.
\end{align*}

To each sub-probability distribution~$\rho$ we associate the function
$f:\Sigma^*\to\bbR$ defined as in~(\ref{eq:1}), and also the following quantities:
\begin{align*}
  \ew&=1-\sum_{a\in\Sigma}\rho_a\,,&Z&=\sum_{x\in\Sigma^*}f(x)=\frac1\ew<\infty\,.
\end{align*}
Finally we define the \emph{sub-Bernoulli measure} $\nu_\rho$ as the
probability distribution over the countable set~$\Sigma^*$, equipped
with the discrete \slgb, and defined by:
\begin{gather*}
  \forall x\in\Sigma^*\quad\nu_\rho(\{x\})=\frac1Zf(x)\,.
\end{gather*}

A \emph{finite Bernoulli sequence} is the random sequence of letters
that compose a word $x\in\Sigma^*$, drawn at random according to the
probability measure~$\nu_\rho$ on~$\Sigma^*$\,.

An \emph{effective} way to produce a finite Bernoulli sequence
according to a sub-probability distribution $p=(p_i)_i$ is the
following. Consider a stopping symbol~$\EOF$, and extend $p$ to a
probability distribution by setting $p(\EOF)=1-\sum_ip_i$. Then output a
Bernoulli sequence according to~$p$, until the symbol $\EOF$ first
occurs. The letters before the first occurrence of $\EOF$ form the sought
finite sequence.

\subsubsection{Full elementary cylinders}
\label{sec:full-elem-cylind}

It is convenient to consider the following completion of~$\Sigma^*$:
\begin{gather*}
\overline{\Sigma^*}=\Sigma^*\cup\B\Sigma^*\,.  
\end{gather*}

Hence both Bernoulli measures and sub-Bernoulli measures are now
defined on the same space~$\overline{\Sigma^*}$.  For each word
$x\in\Sigma^*$, we define the \emph{full elementary cylinder} $\Up x$
as follows:
\begin{gather*}
  \Up x=\{\xi\in\overline{\Sigma^*}\tq x\leq\xi\}\,.
\end{gather*}
Here, $x\leq\xi$ has the same meaning as above if $\xi$ is an infinite
sequence; and it means that $x$ is a prefix of $\xi$ if $\xi$ is a
finite word. We gather the description of both Bernoulli and
sub-Bernoulli measures in the following result.

\begin{theorem}
  \label{thr:1}
  Let $\pr$ be a probability measure on
  $\overline{\Sigma^*}=\Sigma^*\cup\B\Sigma^*$. Assume that the
  function $f:\Sigma^*\to\bbR$ defined by $f(x)=\pr(\Up x)$ is
  positive multiplicative, that is to say, satisfies:
\begin{gather*}
\forall x\in\Sigma^*\quad f(x)>0\\
  \forall x,y\in\Sigma^*\quad f(xy)=f(x)f(y)\,.
\end{gather*}
Define $\rho=(\rho_a)_{a\in\Sigma}$ and $\ew$ by:
\begin{align*}
  \forall a\in\Sigma\quad\rho_a&=f(a)\,,&\ew&=1-\sum_{a\in\Sigma}\rho_a\,.
\end{align*}
Then one and only one of the two following possibilities occurs:
\begin{enumerate}
\item\label{item:5} $\ew=0$. In this case, $\pr$~is concentrated on~$\B\Sigma^*$,
  and characterized by:
  \begin{gather*}
    \forall x\in\Sigma^*\quad\pr(\up x)=\pr(\Up x)=f(x)\,.
  \end{gather*}
The series $\sum_{x\in\Sigma^*}f(x)$ is divergent, and\/ $\pr$ is a
Bernoulli measure.
\item\label{item:6} $\ew>0$. In this case, $\pr$~is concentrated on~$\Sigma^*$.  The
  series $Z=\sum_{x\in\Sigma^*}f(x)$ is convergent, and satisfies:
\begin{align*}
  \ew&=\frac1 Z\,,&
\forall x\in\Sigma^*\quad\pr(\{x\})&=\ew f(x)\,.
\end{align*}
The measure $\pr$ is a sub-Bernoulli measure.
\end{enumerate}
\end{theorem}

\subsection{Bernoulli and sub-Bernoulli measures on trace monoids}
\label{sec:bern-finite-bern-1}

The notions of Bernoulli measure and of sub-Bernoulli measure extend
to trace monoids the notions of Bernoulli sequences and of finite
Bernoulli sequences. They provide a theoretical ground for concurrency
probabilistic models, in the framework of trace monoids.

\subsubsection{Valuations}
\label{sec:valuations}

Let $\bbR_+^*$ be equipped with the monoid structure
$(\bbR_+^*,\times,1)$.  A \emph{valuation} on a trace monoid
$\M=\M(\Sigma,I)$ is a morphism of monoids $f:\M\to\bbR_+^*$. It is
thus a function $f:\M\to\bbR_+^*$ satisfying:
\begin{align*}
  f(\unit)&=1\,,&\forall x,y\in\M\quad f(x\cdot y)&=f(x)f(y).
\end{align*}

\subsubsection{M\"obius transform; M\"obius and sub-M\"obius valuations}
\label{sec:mobius-transform}

Let $f:\M\to\bbR$ be a valuation. The \emph{M\"obius transform} of $f$
is the function $h:\C\to\bbR$ defined by:
\begin{gather}
\label{eq:2}
  \forall c\in\C\quad h(c)=\sum_{c'\in\C\tq c'\geq
    c}(-1)^{|c'|-|c|}f(c'),
\quad\text{only defined on cliques.}
\end{gather}

An alternative expression for the M\"obius transform is the
following. For each clique $c\in\C$, let $\M^{(c)}$ be the sub-trace
monoid generated by those letters $a\in\Sigma$ such that $a\parallel
c$. Here, $a\parallel c$ reads as ``$a$~parallel to~$c$'', and means that
$(a,b)\in I$ for all letters $b$ that occur in the clique~$c$. Then:
\begin{align}
\label{eq:10}
h(c)&=f(c)\cdot \mu_{\M^{(c)}}(t_1,\ldots,t_N)\,,&\text{with }t_i&=f(a_i)\,,
\end{align}
where $\mu_{\M^{(c)}}$ denotes the multivariate M\"obius polynomial
(see Section~\ref{sec:multi-variate-mobius}) of the trace
monoid~$\M^{(c)}$\,. By convention, the expression
$\mu_{\M^{(c)}}(t_1,\ldots,t_N)$ actually involves only the variables
$t_i$ associated with those generators belonging to~$\M^{(c)}$.

In particular, $h(\unit)=\mu_\M(t_1,\ldots,t_N)$---that was already
observable on~(\ref{eq:2}).

A valuation $f$ is said to be~\cite{abbes_mair15}:
\begin{enumerate}
\item A \emph{M\"obius} valuation if:
  \begin{align}
\label{eq:20}
h(\unit)&=0\,,&\forall c\in\Cstar\quad h(c)>0\,.
  \end{align}
\item A \emph{sub-M\"obius} valuation if:
  \begin{align}
\label{eq:21}
h(\unit)&>0\,,&\forall c\in\Cstar\quad h(c)&>0\,.
\end{align}
\end{enumerate}

Equivalently, as seen from~(\ref{eq:10}), if
$\Sigma=\{a_1,\ldots,a_N\}$ and if we put $t_i=f(a_i)$ for
$i\in\{1,\ldots,N\}$, then $f$ is:
\begin{enumerate}
\item A \emph{M\"obius} valuation if:
  \begin{align}
\label{eq:11}
    \mu_\M(t_1,\ldots,t_N)&=0\,,&\forall c\in\Cstar\quad \mu_{\M^{(c)}}(t_1,\ldots,t_N)>0\,.
  \end{align}
\item A \emph{sub-M\"obius} valuation if:
  \begin{align}
\label{eq:19}
    \mu_\M(t_1,\ldots,t_N)&>0\,,&\forall c\in\Cstar\quad \mu_{\M^{(c)}}(t_1,\ldots,t_N)>0\,.
  \end{align}
\end{enumerate}

\subsubsection{Cylinders and \slgb s on $\Mbar$ and on $\BM$}
\label{sec:cylinders}

Recall that we have introduced infinite traces in
Section~\ref{sec:infinite-traces-1}, yielding the completion
$\Mbar=\M\cup\BM$.

To each trace $x\in\M$, we associate the \emph{elementary cylinder}
$\up x\subseteq\BM$ and the \emph{full elementary cylinder}
$\Up x\subseteq\Mbar$, defined as follows:
\begin{align*}
\Up x&=\{\xi\in\Mbar\tq x\leq\xi\}\,,&  \up x&=\{\xi\in\BM\tq
x\leq\xi\}=\Up x\cap\BM\,.
\end{align*}

The set $\Mbar$ is equipped with the \slgb\ $\FFFbar$ generated by the
collection of full elementary cylinders, and the set $\BM$ is equipped
with the \slgb\ $\FFF$ induced by $\FFFbar$ on~$\BM$. Both \slgb s are
Borel \slgb s for compact and metrisable topologies. The restriction of
$\FFFbar$ to $\M$ is the discrete \slgb.

\subsubsection{Multiplicative probability measures}
\label{sec:mult-prob-meas}

In order to state an analogous result to
Theorem~\ref{thr:1} for trace monoids, we introduce the following
definition, borrowed from~\cite{abbes_mair15}.

\begin{definition}
  \label{def:1}
A probability measure\/ $\pr$ on $(\Mbar,\FFFbar)$ is said to be
\emph{multiplicative} whenever it satisfies the following property:
\begin{gather*}
\forall x\in\M\quad\pr(\Up x)>0\,,\\
  \forall x,y\in\M\quad\pr\bigl(\Up(x\cdot y)\bigr)=\pr(\Up x)\cdot\pr(\Up
  y)\,.
\end{gather*}

We define the valuation $f:\M\to\bbR$ associated with\/ $\pr$ and the
number $\ew$ by:
\begin{align*}
  \forall x\in\M\quad f(x)&=\pr(\Up x)\,,&\ew&=h(\unit)\,,
\end{align*}
where $h:\C\to\bbR$ is the M\"obius transform of~$f$. 
\end{definition}

The relationship between multiplicative measures and M\"obius and
sub-M\"obius valuations is as
follows~\cite{abbes_mair15,abbes15_unif}.

\begin{theorem}
  \label{thr:2}
  There is a bijective correspondence between multiplicative
  measures\/~$\pr$ and valuations which are either M\"obius or
  sub-M\"obius.  The alternative is the following:
\begin{enumerate}
\item\label{item:7} $\ew=0$. The valuation is M\"obius. In this case,
  $\pr$~is concentrated on~$\BM$ and characterized by:
  \begin{gather*}
    \forall x\in\M\quad \pr(\up x)=f(x)\,.
  \end{gather*}
  The series $\sum_{x\in\M}f(x)$ is divergent.  We say that\/ $\pr$ is
  a Bernoulli measure.
\item\label{item:8} $\ew>0$. The valuation is sub-M\"obius. In this
  case, $\pr$~is concentrated on~$\M$ and satisfies:
  \begin{gather*}
    \forall x\in\M\quad \pr(\{x\})=\ew f(x)\,.
  \end{gather*}
  The series $Z=\sum_{x\in\M}f(x)$ is convergent and satisfies
  $Z=1/\ew$\,.  We say that\/ $\pr$ is a sub-Bernoulli measure.
\end{enumerate}
\end{theorem}

\subsubsection{Uniform multiplicative measures}
\label{sec:uniform-sub-uniform}

A valuation $f:\M\to\bbR$ is said to be \emph{uniform} if $f(a)$ is
constant, for $a$ ranging over~$\Sigma$. This is equivalent to saying
that $f(x)$ only depends on the length of~$x$, and also equivalent to
saying that $f(x)=p^{|x|}$ for some real~$p$.

A multiplicative measure is \emph{uniform} if the associated valuation
is uniform. The following result describes uniform multiplicative
measures~\cite{abbes_mair15,abbes15_unif}. They are related to the root
of smallest modulus of the M\"obius polynomial of the trace monoid
(see Section~\ref{sec:growth-series-mobius}).

\begin{theorem}
  \label{thr:3}
  Uniform multiplicative measures\/ $\pr$ on a trace monoid $\M$ are
  in bijection with the half closed interval $(0,p_0]$, where $p_0$ is
  the root of smallest modulus of~$\mu_\M$\,. The correspondence
  associates with\/ $\pr$ the unique real $p$ such that\/
  $\pr(\Up x)=p^{|x|}$ for all $x\in\M$.

  The alternative is the following:
\begin{enumerate}
\item $p=p_0$\,. In this case, $\pr$~is Bernoulli (concentrated on the boundary).
\item $p<p_0$\,. In this case, $\pr$~is sub-Bernoulli (concentrated on
  the monoid).
\end{enumerate}
\end{theorem}

\subsubsection{Extension and restriction of valuations}
\label{sec:extens-restr-valu}

We shall need the following result for the construction of the \FSA\
in Section~\ref{sec:full-synchr-algor}.

\begin{theorem}
  \label{thr:4}
Let $\M=\M(\Sigma,I)$ be an irreducible trace monoid.
\begin{enumerate}
\item\label{item:9} Let $f:\M\to\bbR_+^*$ be a M\"obius valuation,
  let\/ $\Sigma'$ be any proper subset of\/~$\Sigma$ and let $\M'$ be
  the submonoid of $\M$ generated by\/~$\Sigma'$. Then the restriction
  $f':\M'\to\bbR_+^*$ of $f$ to $\M'$ is a sub-M\"obius valuation.
\item\label{item:10} Let $\Sigma'=\Sigma\setminus\{a\}$, where $a$ is
  any element of\/~$\Sigma$, let $\M'$ be the submonoid of $\M$ generated
  by~$\Sigma'$, and let $f':\M'\to\bbR_+^*$ be a sub-M\"obius
  valuation. Then there exists a unique M\"obius valuation
  $f:\M\to\bbR_+^*$ that extends $f'$ on~$\M'$.
\end{enumerate}
\end{theorem}

\begin{proof}
  \emph{Proof of point~\ref{item:9}.}\quad Let $f$, $f'$ and $\Sigma'$
  be as in the statement. Let also $h$ and $h'$ denote the M\"obius
  transforms of $f$ and of~$f'$. Let $S$ be the series with
  nonnegative terms:
\begin{gather*}
  S=\sum_{x\in\M'}f(x).
\end{gather*}
We claim that this series is convergent.  To prove it, recall
from~\cite{cartier69} that a pair $(\gamma,\gamma')$ of cliques is
said to be in normal form, denoted by $\gamma\to\gamma'$, if
$\forall b\in\gamma'\ \exists a\in\gamma\ (a,b)\notin I$. Any trace
$x\in\M$ can be uniquely written as a product
$x=\gamma_1\cdot\ldots\cdot \gamma_k$ of cliques such that
$\gamma_i\to\gamma_{i+1}$ holds for all $i=1,\dots,k-1$.

Let the nonnegative matrices
$A=(A_{\gamma,\gamma'})_{(\gamma,\gamma')\in\Cstar\times\Cstar}$ and
$B=(B_{\gamma,\gamma'})_{(\gamma,\gamma')\in\Cstar\times\Cstar}$,
where $\Cstar$ denotes the set of nonempty cliques of~$\M$, be defined
by:
\begin{align*}
  A_{\gamma,\gamma'}&=\un(\gamma\in\M')\cdot\un(\gamma'\in\M')\cdot\un(\gamma\to\gamma')\cdot
                      f(\gamma'),
&B_{\gamma,\gamma'}&=\un(\gamma\to\gamma')\cdot f(\gamma').
\end{align*}

It follows from the existence and uniqueness of the normal form for
traces, decomposing the traces $x\in\M$ according to the number of
terms of their normal form, that the series $S$ writes as:
\begin{gather*}
S=1+\sum_{k\geq1}I\cdot A^k\cdot
  J=1+I\cdot\Bigl(\sum_{k\geq1} A^k\Bigr)\cdot J,
\end{gather*}
where $I$ and $J$ are row and column vectors filled with~$1$s. We know
by \cite[Lemma~6.4]{abbes15} that the matrix $B$ has spectral
radius~$1$, since $f$ is assumed to be M\"obius, and that it is a
primitive matrix since $\M$ is irreducible. But $A\leq B$ with
$A\neq B$. Therefore, it follows from the Perron-Frobenius Theorem
\cite{seneta81} that $A$ has spectral radius $<1$ and thus that the
series $S$ is convergent.

By the M\"obius inversion formula~\cite{viennot86}, it entails that
the relation $\bigl(\sum_{x\in\M'}f(x)\bigr)\cdot h'(\unit)=1$ holds
in the fields of reals. Hence $h'(\unit)>0$. It remains to prove that
$h'(\delta)>0$ also holds for any non empty clique $\delta$
of~$\M'$. This is a bit easier to prove. For any non empty
clique~$\delta$ of~$\M'$, let $\M'^{(\delta)}$ and $\M^{(\delta)}$ be
the submonoids of $\M'$ and of $\M$ respectively, be defined as in
Section~\ref{sec:mobius-transform}. Then
$\M'^{(\delta)}\subseteq\M^{(\delta)}$, hence the following
inequalities between series with nonnegative terms hold:
\begin{gather*}
  \sum_{x\in\M'^{(\delta)}}f(x)\leq\sum_{x\in\M^{(\delta)}}f(x)=\frac1{h(\delta)}<\infty.
\end{gather*}
Therefore, as above, the equality
$h'(\delta)\cdot\bigl(\sum_{x\in\M'^{(\delta)}}f(x)\bigr)=1$ holds in
the field of reals, proving that $h'(\delta)>0$, which completes the
proof that $f'$ is a sub-M\"obius valuation.

\medskip \emph{Proof of point~\ref{item:10}.}\quad Let $a$, $\Sigma'$
and $f':\M'\to\bbR_+^*$ be as in the statement. Let also
$h':\C'\to\bbR$ the M\"obius transform of~$f'$, where
$\C'$ is the set of cliques of~$\M'$.  Assuming
that a M\"obius extension $f:\M\to\bbR_+^*$ of $f'$ exists, we first
prove its uniqueness. For each positive real~$t$, let $f_t:\M\to\bbR$
be the valuation defined by $f_t(a)=t$ and $f_t(\alpha)=f'(\alpha)$
for $\alpha\in\Sigma'$. Then any valuation on $\M$ extending $f'$ is
of the form $f_t$ for some~$t$. Let $h_t:\C\to\bbR$ denote the
M\"obius transform of~$f_t$. We evaluate $h_t(\unit)$ as follows:
\begin{align*}
  h_t(\unit)&=\sum_{\gamma\in\C\tq
              a\in\gamma}(-1)^{|\gamma|}f_t(\gamma)+\sum_{\gamma\in\C\tq a\notin\gamma}(-1)^{|\gamma|}f_t(\gamma).
\end{align*}

On the one hand, observing the equality of sets
$\C'=\{\gamma\in\C\tq a\notin\gamma\}$, we recognize $h'(\unit)$ in
the second sum. On the other hand, using the notation already
introduced $a\parallel\gamma$, for $\gamma\in\C$, to denote that
$a\notin\gamma$ and $a\cdot\gamma\in\C$, the range of the first sum
above is in bijection with the set of cliques $\delta\in\C'$ such that
$a\parallel\delta$, the bijection associating $\delta$ with
$\gamma=a\cdot\delta$. We then have
$(-1)^{|\gamma|}f_t(\gamma)=(-t)\cdot (-1)^{|\delta|}f'(\delta)$.
Henceforth:
\begin{align*}
  h_t(\unit)&=(-t)\cdot K+h'(\unit),&\text{with }K&=1-K'\text{ and
                                                         }K'=\sum_{\delta\in\C'\tq
                                                             a\parallel\delta\wedge
                                                             \delta\neq\unit}(-1)^{|\delta|+1}f'(\delta).
\end{align*}

Let $X$ be a random trace associated with the sub-M\"obius
valuation~$f'$. We evaluate the probability of the event $U=\{\exists
b\in\Sigma'\tq a\parallel b\wedge a\leq X\}$. The inclusion-exclusion
principle yields:
\begin{gather*}
  \pr(U)=\sum_{\delta\in\C'\tq
    a\parallel\delta\wedge\delta\neq\unit}(-1)^{|\delta|+1}\pr(X\geq\delta)=K',
\end{gather*}
the later equality since $\pr(X\geq\delta)=f'(\delta)$ by definition
of~$X$. The event $U$ has probability less than~$1$, otherwise all
pieces of $\Sigma'$ would be parallel to~$a$, contradicting that $\M$
is irreducible. We deduce that $K'\in[0,1)$ and thus $K\in(0,1]$. In
particular, if $f_t$ is M\"obius, it entails that $h_t(\unit)=0$, and
since we have just seen that $K\neq0$, it implies that
$t=h'(\unit)/K$, proving the sought uniqueness.

Let us now prove the existence of a M\"obius extension~$f$. Let
$\M^{(\delta)}$ denote as above, for any clique $\delta\in\C$, the
sub-monoid of $\M$ generated by those letters $b\in\Sigma$ such that
$b\parallel\delta$. For each $\delta\in\C$, we introduce the formal
series:
\begin{gather*}
  G_\delta(t)=\sum_{x\in\M^{(\delta)}}f_t(x).
\end{gather*}

Now, let $t_0$ be the radius of convergence of the power series
$G_{\unit}(t)=\sum_{x\in\M}f_t(x)$.  Since all the power series
$G_\delta(t)$ have non negative coefficients, they satisfy
$G_\delta(t)\leq G_{\unit}(t)<\infty$ for all $t\in[0,t_0)$. In
particular, the radius of convergence $r_\delta$ of $G_\delta$
satisfies $r_\delta\geq t_0$. Actually, reasoning as in the proof of
point~\ref{item:9} and invoking the Perron-Frobenius Theorem, we see
that the strict inequality $r_\delta>t_0$ holds for all
$\delta\neq\unit$. Therefore, for all $\delta\neq\unit$, the series
$G_\delta(t_0)$ is convergent, and thus the M\'obius inversion formula
yields the following equality in the field of reals:
$G_\delta(t_0)\cdot h_{t_0}(\delta)=1$. We conclude that
$h_{t_0}(\delta)>0$ for all cliques $\delta\neq\unit$.  Since
$h_{t_0}(\unit)=0$, we conclude that $f_{t_0}$ is the sought M\"obius
valuation extending~$f'$.
\end{proof}

\section{The Probabilistic Synchronization Algorithm}
\label{sec:basic-synchr-algor}

In this section we consider a network of alphabets sharing some common
letters. We then wish to generate random traces of the synchronization
monoid, in a distributed and incremental way. The first idea that
comes in mind is to generate local Bernoulli sequences, and to see
what is the synchronization trace of these sequences. This constitutes
the Probabilistic Synchronization Algorithm, which is thus a
probabilistic variant of the Generalized Synchronization Algorithm
which was described in Algorithm~\ref{algo3}.

The random traces thus obtained can be either finite or infinite. In
all cases their probability distribution is multiplicative, hence the
theory of Bernoulli and sub-Bernoulli measures for trace monoids
is the adequate tool for their study. After having established this
rather easy result, we turn to specific examples. We obtain the non
trivial result that for path models, any Bernoulli or sub-Bernoulli
measure can be generated by this simple technique.

\subsection{Description of the \BSA}
\label{sec:descr-algor}

Let $(\Sigma_1,\ldots,\Sigma_N)$ be a network of $N$ alphabets with
$N\geq2$, let $\Sigma=\Sigma_1\cup\ldots\cup\Sigma_N$\,, and let
$\M=\M(\Sigma,I)$ be the synchronization trace monoid, as described
in Section~\ref{sec:synchr-alph-synchr}. Recall that we identify $\M$ with
the sub-monoid $\G\subseteq\H$ defined
in Section~\ref{sec:synchr-alph-synchr}, and that indices in $1,\dots,N$
are seen as resources.

Assume that, to each resource $i\in\{1,\ldots,N\}$\,, is attached a
device able to produce a $\Sigma_i$-Bernoulli sequences~$Y_i$, either finite
or infinite, with a specified probability or sub-probability
distribution $p_i$ on~$\Sigma_i$\,. As we have seen in
Section~\ref{sec:finite-bern-sequ}, the generation of such Bernoulli
sequences is effective, together with the information that the
sequence is over in case where $p_i$ is a sub-probability
distribution. 

Henceforth, it is straightforward for each device to produce a
sampling of~$Y_i$ under the form $Z_{i,1}\cdot Z_{i,2}\cdot\ldots$,
and moreover to tag each sub-sequence $Z_{i,k}$ with an additional
symbol $\dag_{i,k}\in\{\EOF,\WFI\}$ in order to deliver a sequence
$(\Zbar_{i,k})_{k\geq1}$ with $\Zbar_{i,k}=Z_{i,k}\cdot\dag_{i,k}$ as
specified in Section~\ref{sec:feed-synchr-algor}. The symbol
$\dag_{i,k}$ is given the value $\WFI$ until the device decides the
sequence is over (if it ever does), after which the symbol
$\dag_{i,k}$ is given the value~$\EOF$.

The Probabilistic Synchronization Algorithm (\BSA), the pseudo-code of
which is given below in Algorithm~\ref{algo4}, consists in executing
in parallel both the local generation of the sequences
$Y_1,\dots,Y_N$, and the Generalized Synchronization Algorithm
described previously in Algorithm~\ref{algo3}.

\begin{algorithm}
\caption{Probabilistic Synchronization Algorithm}
\label{algo4}
\begin{algorithmic}[1]
\Require $p_1,\dots,p_N$
\Comment{Probability or sub-probability distributions}
\While{Algorithm~\ref{algo3} does not exit}
\ForAll {$i\in\{1,\dots,N\}$}
\ForAll{$k=1,2,\dots$}
\State generate the $k^\text{th}$  sampling  $\Zbar_{i,k}$ of a
Bernoulli  sequence according to~$p_i$
\State{feed} Algorithm~\ref{algo3} with
$\Zbar_{i,k}$
\EndFor
\EndFor
\EndWhile
\end{algorithmic}
\end{algorithm}

\subsection{Analysis of the algorithm}
\label{sec:analysis-algorithm}

\subsubsection{Distribution of\/ \BSA\ random traces}
\label{sec:qualitative-result}

The trace $y\in\Mbar$, output of the execution of the \BSA\
(Algorithm~\ref{algo4}), is random. What is its distribution?

\begin{theorem}
  \label{thr:5}
  The probability distribution\/ $\pr$ of the random trace produced by the \BSA\ is a
  multiplicative probability measure on~$\Mbar$. The valuation
  $f:\M\to\bbR$ associated with this measure by $f(x)=\pr(\Up x)$ is
  such that:
  \begin{gather}
    \label{eq:9}
\forall a\in\Sigma\quad f(a)=\prod_{i\in R(a)}p_i(a)\,,
  \end{gather}
  where $p_i$ is the probability or sub-probability distribution
  on~$\Sigma_i$\,, and $R(a)$ is the set of resources associated
  with~$a$.
\end{theorem}

\begin{proof}
  For each letter $a\in\Sigma$, let $q_a$ be the real number defined
  by:
\begin{gather*}
  q_a=\prod_{i\in R(a)}p_i(a)\,.
\end{gather*}

Let $y\in\Mbar$ be the random trace produced by the \BSA. Fix $z\in\M$
a trace, and let $(z_1,\ldots,z_N)$ be the representation of $z$
in~$\G$.  Then $y\geq z$ holds if and only if, for every index
$i\in\{1,\ldots,N\}$, the corresponding sequence $x_i$ starts
with~$z_i$\,.

Let $z=a_1\cdot\ldots\cdot a_j$ be any decomposition of $z$ as a
product of generators $a_k\in\Sigma$. Since each sequence $x_i$
produced by the device number~$i$ is Bernoulli or sub-Bernoulli with
distribution or sub-distribution~$p_i$\,, and since the sequences are
mutually independent, we have thus:
  \begin{align*}
    \pr(y\geq z)&=\prod_{k=1}^jq_{a_k}\,.
  \end{align*}

  In other words, if $f:\M\to\bbR$ is the valuation defined
  by~(\ref{eq:9}), one has $\pr(\Up z)=f(z)$. This shows that $\pr$ is
  multiplicative.
\end{proof}

\subsubsection{Small values}
\label{sec:small-values}

We shall see that not all multiplicative probabilities on $\Mbar$ can
be reached by this technique. However, as long as only ``small
values'' are concerned, the \BSA\ can reach any target multiplicative
measure. Let us first introduce the following convention: if
$f:\M\to\bbR$ is a valuation, and if $\ew$ is a real number, we denote
by $\varepsilon f$ the valuation which values on generators are given
by:
\begin{gather*}
  \forall a\in\Sigma\quad (\ew f)(a)=\ew f(a)\,.
\end{gather*}
Hence the value of $\ew f$ on arbitrary traces are given by:
\begin{gather*}
  \forall x\in\M\quad (\ew f)(x)=\ew^{|x|}f(x)\,.
\end{gather*}

The following result states that, at the expense of having small
synchronizing traces, the relative frequency of letters in traces
produced by the \BSA\ suffers from no constraint.

\begin{proposition}
  \label{prop:1}
  Let $(\Sigma_1,\ldots,\Sigma_N)$ be $N$ alphabets, let
  $\M=\M(\Sigma,I)$ be the synchronization trace monoid, and let
  $t:\M\to\bbR$ be a valuation on~$\M$. Then there exists $\ew>0$ and
  sub-probability distributions $(p_i)_{1\leq i\leq N}$\,, with $p_i$
  a sub-probability distribution on~$\Sigma_i$\,, such that the
  valuation $f$ characterizing the random trace $y\in\M$ produced by
  the \BSA\ with respect to $(p_i)_{1\leq i\leq N}$\,, satisfies:
  \begin{gather*}
    f=\ew t\,.
  \end{gather*}
\end{proposition}

\begin{proof}
  For $\ew>0$ to be specified later, let $(p_i)_{1\leq i\leq N}$ be
  the family of real valued functions, $p_i:\Sigma_i\to\bbR$\,,
  defined by:
  \begin{gather*}
    \forall a\in\Sigma_i\quad p_i(a)=\bigl(\ew t(a)\bigr)^{1/|R(a)|}\,.
  \end{gather*}

  For $\ew>0$ small enough, all $p_i$ are sub-probability
  distributions over~$\Sigma_i$\,. According to Th.~\ref{thr:5}, the
  valuation $f$ describing the distribution of the random trace produced by the
  \BSA\ satisfies:
  \begin{align*}
    f(a)&=\prod_{i\in R(a)}p_i(a)=\ew t(a)\,.
  \end{align*}
The proof is complete.
\end{proof}

\subsubsection{Reduction to irreducible trace monoids}
\label{sec:reduct-irred-trace}

Assume that the trace monoid is not irreducible. Hence the dependence
relation $(\Sigma,D)$ has several connected components, let us say
that it has two components $(S_1,D_1)$ and $(S_2,D_2)$ to simplify the
discussion. 

Let $I_1$ and $I_2$ be the dependence relations on $S_1$ and $S_2$
defined by:
\begin{align*}
  I_1&=(S_1\times S_1)\cap I=(S_1\times S_1)\setminus D_1\,,&
  I_2&=(S_2\times S_2)\cap I=(S_2\times S_2)\setminus D_2\,.&
\end{align*}

Putting $\M_1=\M(S_1,I_1)$ and $\M_2=\M(S_2,I_2)$, we have:
\begin{align*}
  \M&=\M_1\times\M_2\,,&\Mbar&=\overline{\M_1}\times\overline{\M_2}\,.
\end{align*}

It follows from Theorem~\ref{thr:5} that the distribution $\pr$ of the random
trace $y$ produced by the \BSA\ is a tensor product
$\pr_1\otimes\pr_2$; probabilistically speaking, $y$~is obtained as
the concatenation of two independent traces $y_1\in\M_1$ and
$y_2\in\M_2$\,.  The probability measures $\pr_1$ and~$\pr_2$\,, of $y_1$ and $y_2$
respectively, are identical as those deriving from \BSA\ algorithms
restricted to the alphabets concerning $S_1$ and $S_2$ respectively.

In conclusion, the \BSA\ algorithm decomposes as a product of sub-\BSA\
algorithms on the irreducible components of the synchronization trace
monoid. Hence there is no loss of generality in assuming, for the sake
of analysis, that the synchronization trace monoid is irreducible.

\subsection{Example: the path model}
\label{sec:examples}

The path model is close to the dimer model, a topic of numerous
studies in Combinatorics and in Statistical
Physics~\cite{viennot85,khanin10}. Here we shall see that the path
model is a framework where the \BSA\ works at its best, in the sense
that any multiplicative measure on $\Mbar$ can be obtained through the
\BSA.

The path model is defined as the trace monoid $\M=\M(\Sigma,I)$ on
$N+1$ generators: $\Sigma=\{a_0,\ldots,a_N\}$, where $N\geq0$ is a
fixed integer. The independence relation is defined by:
\begin{gather*}
  I=\{(a_i,a_j)\tq |i-j|>1\}\,.
\end{gather*}

Then $\M$ is the synchronization monoid of the $N$-tuple of alphabets
$(\Sigma_1,\ldots,\Sigma_N)$ with $\Sigma_i=\{a_{i-1},a_i\}$ for
$i\in\{1,\ldots,N\}$. The set of resources of $a_i$ is:
\begin{gather}
\label{eq:7}
  R(a_i)=
  \begin{cases}
\{1\},&\text{if   $i=0$},\\
\{N\},&\text{if $i=N$},\\
\{i,i+1\},&\text{if $0<i<N$.}
  \end{cases}
\end{gather}

Assume that each of the $N$ alphabets $\Sigma_1,\ldots,\Sigma_N$ is
equipped with a positive probability distribution, say
$p_1,\ldots,p_N$\,. Then, with probability~$1$, all tuples
$(x_1,\ldots,x_N)$, where $x_i$ is a Bernoulli infinite sequence
distributed according to~$p_i$\,, are synchronizing. Indeed, each
$x_i$ contains infinitely many occurrences of $a_{i-1}$ and
of~$a_i$\,, and this is enough to ensure the
synchronization. Therefore, the \BSA\ yields a multiplicative measure
entirely supported by the set of infinite traces of~$\M$, hence a
Bernoulli measure on~$\BM$.

Actually, the following result shows that \emph{every} Bernoulli
measure on the path model can be obtained through the execution of
the~\BSA.

\begin{theorem}
  \label{thr:6}
  Let\/ $\pr$ be a Bernoulli measure on~$\BM$, where $\M$ is the
  synchronization monoid associated with the path model. For
  $i\in\{0,\ldots,N\}$, put:
  \begin{gather*}
    t_i=\pr(\up a_i)\,.
  \end{gather*}

  Then there exists a unique tuple $(p_1,\ldots,p_N)$, with $p_i$ a
  positive probability distribution over $\Sigma_i=\{a_{i-1},a_i\}$\,,
  such that the random infinite trace $\xi\in\BM$ produced by the
  \BSA\ based on $(p_1,\ldots,p_N)$ has the distribution
  probability~$\pr$.

  The probability distributions $p_1,\ldots,p_N$ are computed
  recursively by:
\begin{align}
\label{eq:5}
  p_1(a_0)&=t_0&p_1(a_1)&=1-t_0\\
\label{eq:6}
i\in\{2,\ldots,N\}\quad  p_i(a_{i-1})&=\frac{t_{i-1}}{p_{i-1}(a_{i-1})}&p_i(a_i)&=1-p_i(a_{i-1})\,.
\end{align}
\end{theorem}

\begin{remark}
  The relations~(\ref{eq:5})--(\ref{eq:6}) are necessary conditions,
  almost immediate to establish from the result of
  Theorem~\ref{thr:5}. What is not obvious is that the numbers thus
  defined stay within $(0,1)$, yielding indeed probability
  distributions $\begin{pmatrix}p_i(a_{i-1})&p_i(a_i)\end{pmatrix}$ 
  over~$\Sigma_i$\,, and that $p_N(a_N)=t_N$.
\end{remark}

\begin{proof}
  Let $f:\M\to\bbR$ be the valuation associated with~$\pr$. According
  to Theorem~\ref{thr:2}, $f$~is a M\"obius valuation---we will use
  this fact in a moment.

  \emph{Proof of uniqueness of\/ $(p_1,\ldots,p_N)$ and proof
    of\/~\eqref{eq:5}--\eqref{eq:6}.}\quad With respect to a tuple
  $(p_1,\ldots,p_N)$ of positive probability distributions over
  $(\Sigma_1,\ldots,\Sigma_N)$\,, the \BSA\ produces a random infinite
  trace $\xi\in\BM$. Let $g:\M\to\bbR$ be the valuation associated
  with the probability distribution of~$\xi$.  Then, according to
  Theorem~\ref{thr:5}, one has:
  \begin{gather*}
\forall i\in\{0,\ldots,N\}\quad g(a_i)=\prod_{r\in R(a_i)}p_r(a_i)\,.
  \end{gather*}

  Referring to~(\ref{eq:7}), $f=g$ is equivalent to:
\begin{align}
\label{eq:14}
  t_0&=p_1(a_0)\\
\label{eq:15}
t_N&=p_N(a_N)\\
\label{eq:16}
0<i<N\qquad t_i&=p_i(a_i)p_{i+1}(a_i)
\end{align}

It follows at once that the tuple $(p_1,\ldots,p_N)$ inducing~$\pr$
through the \BSA, if it exists, is unique and satisfies necessarily
the recurrence relations~(\ref{eq:5})--(\ref{eq:6}).

\medskip \emph{Proof of existence of $(p_1,\ldots,p_N)$.}\quad Instead
of starting from the recurrence relations~(\ref{eq:5})--(\ref{eq:6}),
we use a different formulation. For $i\in\{-1,\ldots,N\}$, let
$\M_{0,i}$ be the sub-trace monoid of $\M$ generated by
$\{a_0,\ldots,a_i\}$. Let also $\mu_{0,i}$ be the evaluation on
$(t_0,\ldots,t_i)$ of the multivariate M\"obius polynomial
of~$\M_{0,i}$ (see Section~\ref{sec:multi-variate-mobius}). Hence:
\begin{align*}
\mu_{0,-1}&=1\\
  \mu_{0,0}&=1-t_0\\
\mu_{0,1}&=1-t_0-t_1\\
\mu_{0,2}&=1-t_0-t_1-t_2+t_0t_2\\
\mu_{0,3}&=1-t_0-t_1-t_2-t_3+t_0t_2+t_0t_3+t_1t_3\\
\mu_{0,4}&=1-t_0-t_1-t_2-t_3-t_4+t_0t_2+t_0t_3+t_0t_4+t_1t_3+t_1t_4+t_2t_4-t_0t_2t_4\\
&\cdots
\end{align*}

For any $i\in\{-1,\ldots,N-2\}$, the monoid $\M_{0,i}$ coincides with
the sub-monoid $\M^{(c_i)}$ as defined
in Section~\ref{sec:mobius-transform}, where $c_i$ is the following clique
of~$\M$:
\begin{gather*}
  c_i=\begin{cases} a_{i+2}\cdot a_{i+4}\cdot\ldots\ldots\cdot
    a_N\,,&\text{if $N-i\equiv 0\mod 2$,}\\
    a_{i+2}\cdot a_{i+4}\cdot\ldots\ldots\cdot a_{N-1}\,,&\text{if
      $N-i\equiv1\mod 2$.}
  \end{cases}
\end{gather*}

The clique $c_i$ is non empty as long as $i<N-1$.  Therefore,
according to~(\ref{eq:11}), the M\"obius conditions on $f$ ensure the
positivity of all numbers~$\mu_{0,i}$ for $i<N-1$.

Let $i\in\{0,\dots,N-2\}$. Any clique $\gamma$ of $\M_{0,i+1}$ either
belongs to~$\M_{0,i}$, or contains an occurrence of~$a_{i+1}$. In the
later case, this clique $\gamma$ is of the form
$\gamma=a_{i+1}\cdot\gamma'$, with $\gamma'$ ranging over cliques
of~$\M_{0,i-1}$. It follows at once that the following recurrence
relation holds:
\begin{align}
  \label{eq:13}
0\leq i<N-1\quad \mu_{0,i+1}&=\mu_{0,i}-t_{i+1}\mu_{0,i-1}
\end{align}

Similarly, any clique $\gamma$ of~$\M_{0,N}$ is either contained
in~$\M_{0,N-2}$\,, or it contains an occurrence of~$a_{N-1}$, in which
case it writes as $\gamma=a_{N-1}\cdot\gamma'$ with $\gamma'$ ranging
over cliques of~$\M_{0,N-3}$, or it contains occurrence of~$a_N$, in
which case it writes as $\gamma=a_N\cdot\gamma'$ with $\gamma'$
ranging over cliques of~$\M_{0,N-2}$. We deduce:
$\mu_{0,N}=\mu_{0,N-2}-t_{N-1}\mu_{0,N-3}-t_N\mu_{0,N-2}$.  But
$\mu_{0,N}$~is the evaluation on $(t_0,\dots,t_N)$ of the multivariate
M\"obius polynomial of $\M=\M_{0,N}$, hence $\mu_{0,N}=0$ since $f$ is
a M\"obius valuation. We obtain:
\begin{gather}
  \label{eq:3}
(1-t_N)\mu_{0,N-2}-t_{N-1}\mu_{0,N-3}=0.
\end{gather}

Now, since all the numbers $\mu_{0,i}$ for $i\in\{-1,\dots,N-2\}$ are
positive, we define the family $(p_i)_{1\leq i\leq N}$ as follows:
\begin{align}
  \label{eq:12}
p_1(a_0)&=t_0&p_1(a_1)&=1-t_0\\
\label{eq:17}
1<i< N\quad p_i(a_{i-1})&=t_{i-1}\frac{\mu_{0,i-3}}{\mu_{0,i-2}}&
p_i(a_i)&=\frac{\mu_{0,i-1}}{\mu_{0,i-2}}\\
\label{eq:18}
p_N(a_{N-1})&=1-t_N&p_N(a_N)&=t_N
\end{align}

All numbers appearing in~(\ref{eq:12}),
(\ref{eq:17})~and~(\ref{eq:18}) are positive, and
equations~(\ref{eq:14}) and~(\ref{eq:15}) are satisfied. As
for~(\ref{eq:16}), we write, for $1<i<N-1$:
\begin{align*}
p_i(a_i)p_{i+1}(a_i)&=\frac{\mu_{0,i-1}}{\mu_{0,i-2}}t_i\frac{\mu_{0,i-2}}{\mu_{0,i-1}}=t_i
\end{align*}
For $i=N-1$, we have:
\begin{align*}
  p_{N-1}(a_{N-1})p_N(a_{N-1})&=\frac{\mu_{0,N-2}}{\mu_{0,N-3}}(1-t_N)=t_{N-1},
\end{align*}
the later equality by~(\ref{eq:3}). We have shown so far
that~(\ref{eq:14}), (\ref{eq:15})~and~(\ref{eq:16}) are satisfied.

It remains to see that all
$\left(\xymatrix@1@C=1em{ p_i(a_{i-1})&p_i(a_i)}\right)$\,, for $i$
ranging over $\{2,\ldots,N-1\}$, are positive probability vectors,
since this is trivially true for $i=1$ and for $i=N$. We have already observed that
they are all positive vectors. For $1<i<N$, one has:
\begin{align*}
  p_i(a_i)+p_i(a_{i-1})&=\frac{\mu_{0,i-1}+t_{i-1}\mu_{0,i-3}}{\mu_{0,i-2}}=1
\end{align*}
by virtue of~(\ref{eq:13}). Hence each $p_i$ is indeed a positive
probability distribution on~$\Sigma_i$\,, which completes the proof.
\end{proof}

Theorem~\ref{thr:6} can be adapted to the case of sub-Bernoulli
measures instead of Bernoulli measures, still for the path model, as
follows. In a nutshell: every sub-Bernoulli measure can be simulated
by synchronization of sub-Bernoulli sequences, but there is no
uniqueness in the choice of the local sub-probability distributions.

\begin{theorem}
  \label{thr:8}
  Let\/ $\pr$ be a sub-Bernoulli measure on~$\M$, where $\M$ is the
  synchronization monoid associated with the path model. For
  $i\in\{0,\ldots,N\}$, put:
  \begin{gather*}
    t_i=\pr(\up a_i)\,.
  \end{gather*}

  Then there exists a tuple $(p_1,\ldots,p_N)$, with $p_i$ either a
  probability or a sub-probability distribution over
  $\Sigma_i=\{a_{i-1},a_i\}$\,, such that the random trace
  $\xi\in\Mbar$ produced by the \BSA\ based on $(p_1,\ldots,p_N)$ is
  finite with probability~$1$ and has the distribution\/~$\pr$.
\end{theorem}

\begin{proof}
  Using the same notations as in the proof of Theorem~\ref{thr:6}, we
  define $(p_i)_{1\leq i\leq N}$ as follows:
\begin{align}
\label{eq:32}
p_1(a_0)&=t_0&p_1(a_1)&=1-t_0\\
\label{eq:33}
1<i< N\quad p_i(a_{i-1})&=t_{i-1}\frac{\mu_{0,i-3}}{\mu_{0,i-2}}&
p_i(a_i)&=\frac{\mu_{0,i-1}}{\mu_{0,i-2}}\\
\label{eq:34}
p_N(a_{N-1})&=\frac{t_{N-1}}{p_{N-1}(a_{N-1})}&p_N(a_N)&=t_N
\end{align}

The only difference with the definitions introduced in the proof of
Theorem~\ref{thr:6} lies in~(\ref{eq:34}). As a consequence, all
$p_i(a_{i-1})$ and $p_i(a_i)$ are positive. They satisfy
$p_i(a_i)p_{i+1}(a_i)=t_i$ for all $i\in\{2,\ldots,N-1\}$, and
obviously $p_1(a_0)=t_0$ and $p_N(a_N)=t_N$\,. What remains to be
proved is that the sums $p_i(a_{i-1})+p_i(a_i)$ stay within $(0,1]$
for all $i\in\{1,\ldots,N\}$.

For $i\in\{1,\ldots,N-1\}$, the sum is~$1$, just as in the proof of
Theorem~\ref{thr:6}. And for $i=N$, we have:
\begin{align*}
  p_N(a_{N-1})+p_N(a_N)\leq1&\iff
  t_N+t_{N-1}\frac{\mu_{0,N-3}}{\mu_{0,N-2}}\leq1\\
&\iff \mu_{0,N-2}-t_N\mu_{0,N-2}-t_{N-1}\mu_{0,N-3}\geq0\\
&\iff \mu_\M(t_0,\ldots,t_N)\geq0
\end{align*}
The last condition is satisfied since $\pr$ is a sub-Bernoulli
measure. The proof is complete.
\end{proof}

\subsection{Finitary cases}
\label{sec:finitary-cases}

\subsubsection{Example of a ring model}
\label{sec:ring-example}

Consider the ring model already introduced in
Section~\ref{sec:illustrating-bsa-2}. The trace monoid is:
\begin{gather*}
  \M=\langle a_0,a_1,a_2,a_3\;|\; a_0a_2=a_2a_0,\ a_1a_3=a_3a_1\rangle.
\end{gather*}
This is the synchronization monoid associated with the network of
alphabets $(\Sigma_0,\Sigma_1,\Sigma_2,\Sigma_3)$ given by
$ \Sigma_0=\{a_3,a_0\}$, $\Sigma_1=\{a_0,a_1\}$,
$\Sigma_2=\{a_1,a_2\}$ and $\Sigma_3=\{a_2,a_3\}$.  Contrasting with
the path model, we shall see on an example that the \BSA\ for this
ring model with four generators produces finite traces with
probability~$1$.

Let $p_0,\ldots,p_3$ be uniform distributions on
$\Sigma_0,\ldots,\Sigma_3$\,.  According to Theorem~\ref{thr:5}, the
\BSA\ yields a multiplicative and uniform probability measure $\pr$ on
$\Mbar$ with parameter $p=(1/2)(1/2)=1/4$:
\begin{gather}
  \label{eq:8}
  \forall x\in\M\quad\pr(\Up x)=\Bigl(\frac14\Bigr)^{|x|}\,.
\end{gather}

The M\"obius polynomial of $\M$ is $\mu_\M(t)=1-4t+2t^2$\,, which root
of smallest modulus is $p_0=1-1/\sqrt2$. Since $1/4<p_0$, it follows
from Theorem~\ref{thr:3} that the \BSA\ produces a finite trace
$y\in\M$ with probability~$1$. 

The average length of $y$ is easily computed. Following the notations
introduced in Section~\ref{sec:bern-finite-bern-1}, we put
 $G(z)=\sum_{x\in\M}z^{|x|}$, $p=1/4$ and $\varepsilon=\mu_\M(p)$.
  Using that $\pr(y=x)=\varepsilon p^{|x|}$ and $G(z)=1/\mu_\M(z)$, we have:
\begin{gather*}
  \esp|y|=\sum_{x\in\M}|x|\pr(y=x)=\varepsilon p\frac{dG}{dz}\Bigr|_{[z=p]}=
  -p\frac{\mu_\M'(p)}{\mu_\M(p)}=6.
\end{gather*}




\subsubsection{Generalization: trees and cycles}
\label{sec:gener-trees-cycl}

Since the probability distribution of the random element produced by
the \BSA\ is a multiplicative probability measure, we know by
Theorem~\ref{thr:2} that the trace is either finite with
probability~$1$ or infinite with probability~$1$. The examples studied
above show that both cases may occur indeed.

It turns out that the dichotomy is solved by a simple criterion, as
stated below.

\begin{proposition}
  \label{prop:3}
  Assume that the synchronization monoid $\M=\M(\Sigma,I)$ associated
  with a tuple of alphabets $(\Sigma_1,\ldots,\Sigma_N)$ is
  irreducible, and let $D$ be the dependence relation of~$\M$. Let
  $p_1,\ldots,p_N$ be positive probability distributions on
  $\Sigma_1,\ldots,\Sigma_N$ respectively. Then
  \begin{enumerate}
  \item\label{item:3} The \BSA\ based on $(p_1,\ldots,p_N)$ produces
    an infinite trace with probability~$1$ if and only if:
  \begin{enumerate}
  \item\label{item:1} $i\neq j\implies |\Sigma_i\cap\Sigma_j|\leq1$; and
  \item\label{item:2} the non oriented graph $(\Sigma,D)$ has no cycle.
  \end{enumerate}
\item\label{item:4} If the output $X$ of the \BSA\ is finite with
  probability~$1$, then the length $|X|$ has a finite average.
  \end{enumerate}
\end{proposition}

\begin{proof}
  Assume first that both conditions~\ref{item:1} and~\ref{item:2} are
  met. Let $Y=(Y_i)_{i\in\Sigma}$ be a tuple of infinite sequences
  $Y_i\in(\Sigma_i)^\omega$.  With probability~$1$, for any distinct
  $i$ and~$j$ such that $\Sigma_i\cap\Sigma_j\neq\emptyset$, the
  sequence $Y_i$~has infinitely many occurrences of the unique element
  belonging to $\Sigma_i\cap\Sigma_j$. Just as in the case of the path
  model, this is enough to guarantee that the synchronization trace of
  $Y$ is infinite.

  Conversely, assume that one of conditions~\ref{item:1}
  and~\ref{item:2} is not met, for instance
  condition~\ref{item:1}. Let $i$ and $j$ be distinct indices such
  that $\Sigma_i\cap\Sigma_j$ has at least two distinct elements $a$
  and~$b$. Let $Y=(Y_r)_{r\in\Sigma}$ be a random vector of sequences
  such that the synchronization trace of $Y$ is infinite. Since $\M$
  is assumed to be irreducible, we observe that, with probability~$1$,
  if the synchronization trace of $Y$ is infinite then all coordinates
  of $Y$ are infinite. Hence there is no loss of generality in
  assuming that both coordinate $Y_i$ and $Y_j$ are infinite. Then the
  order of occurrences of $a$ and $b$ in both coordinates must be the
  same, which has probability~$0$ to occur.

  Finally, assume that condition~\ref{item:2} is not met, hence the
  presence of a cycle $(r_1,r_2,\dots,r_k)$ in the dependence
  relation. We assume without loss of generality that $r_1,\dots,r_{k-1}$
  are pairwise disjoint. Let $a_1\in\Sigma_{r_1}\cap\Sigma_{r_2}$,
  $a_2\in\Sigma_{r_2}\cap\Sigma_{r_3},\dots$,
  $a_k\in\sigma_{r_k}\cap\Sigma_{r_1}$. Focusing on the coordinates
  $r_1,\dots,r_k$ of $Y$ only, a pattern of the form
  $(a_1a_k,a_2a_1,a_3a_2,\dots,a_{k-1}a_{k-2},a_ka_{k-1})$ shall occur
  with probability~$1$. Since such a pattern is blocking the
  synchronization, it follows that the synchronization trace is finite
  with probability~$1$.

We have proved so far the equivalence stated in point~\ref{item:3}. We
now come to the proof of point~\ref{item:4}, and assume that $|X|<\infty$
with probability~$1$. The average of $|X|$ is computed as the
following mathematical expectation:
\begin{gather*}
  \esp|X|=\sum_{x\in\M}|x|\pr(X=x)=\varepsilon\sum_{x\in\M\setminus\{\unit\}}|x|f(x),
\end{gather*}
where $f$ is the sub-M\"obius valuation associated with~$\pr$, and
$\varepsilon$ is the constant given by the M\"obius transform of $f$
evaluated at the empty heap (see Section~\ref{sec:mult-prob-meas}).
Let $G(\lambda)$ be the power series:
\begin{gather*}
  G(\lambda)=\sum_{x\in\M}\lambda^{|x|}f(x)
\end{gather*}
Then $G(1)=\varepsilon^{-1}<\infty$ according to Theorem~\ref{thr:1}
point~\ref{item:6}, and since $G$ has non negative coefficients, it
implies by the Pringsheim Theorem that the radius of convergence of $G$ is
greater than~$1$. Hence so does its derivative, and thus
$G'(1)<\infty$. Since $\esp|X|=\varepsilon G'(1)$, the result of
point~\ref{item:4} follows.
\end{proof}

\section{The Probabilistic Full Synchronization Algorithm}
\label{sec:iter-meas-full}

The \BSA\ produces random traces, either finite or infinite. 
Proposition~\ref{prop:3} shows that the ability of the \BSA\ to
produce finite or infinite traces does not depend on the probabilistic
parameters one chooses to equip the local alphabets with. It rather depends
on the structure of the synchronization monoid.

To be sound, testing procedures and statistical averaging techniques
require arbitrary large traces, which binds us to the mathematical
model of infinite traces. In case where the \BSA\ fails to produce
infinite traces, we are thus left with an unsolved problem. Yet, we
can produce finite traces\ldots\ and the most natural thing to try
from there, is to start the \BSA\ over and over, and to concatenate
the finite random traces obtained at each execution of the \BSA. The
limiting trace is infinite and random, couldn't it just be the one we
were looking for?

It is quite surprising to realize that this strategy fails in
general. Of course the unlimited concatenation of finite traces, with
a positive average length, necessarily produces an infinite trace. But
the failure comes from the \emph{distribution} of this random infinite
trace. We will show on an example below that it is \emph{not} a
Bernoulli measure in general. In particular, the uniform measure is
thus unreachable by this technique.

Nevertheless, we introduce an algorithm based on recursive
concatenation of finite random traces and that outputs, if executed
indefinitely, an infinite trace which is always distributed according
to a Bernoulli measure. The intermediate, finite random traces,
are obtained by a trial-and-reject procedure based on the \BSA. The
whole procedure constitutes the Probabilistic Full Synchronization
Algorithm (\FSA). In Section~\ref{sec:examples-1}, we show that any
Bernoulli measure can be simulated by the output of the \FSA\ for the
ring model.

\subsection{Convolution of probability distributions and random walks}
\label{sec:conv-prob-distr}

\subsubsection{Convolution and random walks}
\label{sec:conv-rand-walks}

We recall the general definition of convolution for a countable
monoid~$\M$. Let $\nu$ and $\theta$ be two probability distributions
on~$\M$. Assume that $X$ and $Y$ are two independent random variables
with values in~$\M$, distributed according to $\nu$ and to $\theta$
respectively. Then the \emph{convolution $\nu*\theta$} is the distribution of
the random variable~$X\cdot Y$, and it is given by the Cauchy product
formula:
\begin{gather*}
\forall x\in\M\quad  \nu*\theta(x)=\sum_{(y,z)\in\M\tq y\cdot z=x}\nu(y)\cdot\theta(z).
\end{gather*}
The convolution product is associative.

Given a probability distribution $\nu$ over~$\M$, let $(X_n)_{n\geq1}$
be a sequence of independent random variables with values in~$\M$, and
identically distributed according to~$\nu$. The \emph{random walk}
associated with $\nu$ is the sequence of random variables
$(Y_n)_{n\geq0}$ defined by $Y_0=\unit$, the unit element of the
monoid, and inductively: $Y_{n+1}=Y_n\cdot X_{n+1}$ for all
integers~$n$. If $\nu_n$ denotes the distribution of~$Y_n$, we have
$\nu_n=\nu^{*(n)}$ for all integers~$n$, the $n^\text{th}$
convolution power of $\nu$ with itself.

Each trajectory $(Y_n)_{n\geq0}$ of the random walk is nondecreasing
for the divisibility relation in the monoid. Hence, if we assume now
that $\M$ is a trace monoid, the nondecreasing sequence
$(Y_n)_{n\geq1}$ has a least upper bound in the completion~$\Mbar$,
say $Y_\infty=\bigvee_{n\geq1}Y_n$. We introduce the notation
$\nu^{*\infty}$ for the probability distribution of~$Y_\infty$, which we call
the \emph{limit distribution of the random walk} (it is also called the
harmonic measure of the random walk).

\subsubsection{An example where the concatenation of\/ \BSA\ traces does
  not yield a limit multiplicative measure}
\label{sec:negative-result}

We consider the ring synchronization monoid on four generators
$\M=\langle a_0,a_1,a_2,a_3\;|\; a_0a_2=a_2a_0,\ a_1a_3=a_3a_1\rangle$,
corresponding to the network of alphabets
$(\Sigma_0,\Sigma_1,\Sigma_2,\Sigma_3)$ with $\Sigma_0=\{a_3,a_0\}$,
$\Sigma_1=\{a_0,a_1\}$, $\Sigma_2=\{a_1,a_2\}$ and
$\Sigma_3=\{a_2,a_3\}$.  Each alphabet~$\Sigma_i$\,, for
$i=0,\ldots,3$, is equipped with the uniform probability
distribution. Let $\nu^{*\infty}$ be the distribution on $\BM$ of the infinite
trace obtained by concatenating infinitely many independent copies of
a finite trace generated by the \BSA. Then we claim: \emph{the limit
  distribution $\nu^{*\infty}$ is not Bernoulli.}

Seeking a contradiction, assume that it is. Clearly,
$\nu^{*\infty}$~is concentrated on the boundary~$\BM$. And for
symmetry reasons, it is necessarily the uniform distribution,
and thus given by $\nu^{*\infty}(\up x)=p_0^{|x|}$ where $p_0$ is
the root of smallest modulus of the M\"obius polynomial
$\mu(z)=1-4z+2z^2$.

We extend the concatenation of traces $x\cdot y$ with
$(x,y)\in\M\times\M$ to the case where $y$ is an infinite trace by
putting $x\cdot y=\bigvee\{x\cdot y_n\tq n\geq1\}$ for
$y=\bigvee\{y_n\tq n\geq1\}$, and this definition does not depend on
the choice of the sequence~$(y_n)_{n\geq1}$. This yields also an
extension of the notion of convolution $\nu*\theta$ to the case where
$\nu$ is concentrated on~$\M$, but $\theta$ might be a probability
distribution on~$\Mbar$. The construction of $\nu^{*\infty}$ implies
the fix point property $\nu_p*\nu^{*\infty}=\nu^{*\infty}$, where
$\nu_p(\Up x)=p^{|x|}$ is the distribution of the \BSA, here given by $p=1/4$.
In particular for the cylinder $\up a_0$\,, this yields:
\begin{align*}
  p_0&=\sum_{k\geq0}\nu_p(\{a_2^k\})\nu^{*\infty}(\up a_0)+\nu_p(\Up
  a_0)=p_0(1-4p+2p^2)\frac1{1-p}+p
\end{align*}

Simplifying by $p\neq0$, we obtain: $ p=(3p_0-1)/(2p_0-1)$, and since
$1-4p_0+2p_0^2=0$, it yields $p=p_0$, a contradiction. Actually, we
have shown the strongest result that no random walk based on the
distributions $\nu_p(\Up x)=p^{|x|}$ with $p\in(0,p_0)$ has a
Bernoulli measure as limit distribution~$\nu_p^{*\infty}$.

\subsection{First hitting times and pyramidal heaps}
\label{sec:first-hitting-times-2}

\subsubsection{First hitting times}
\label{sec:first-hitting-times-3}

First hitting times for random heaps formalize the idea of the first
time of occurrence of a given piece---yet, without an explicit notion
of time at hand. It generalizes to random heaps the analogous notion,
for a Bernoulli sequence, of first time of reaching a given letter.

Let $\M=\M(\Sigma,I)$ be a trace monoid, and let $a\in\Sigma$ be a
given letter. The number of occurrences of $a$ in the congruent words
defining a trace $x$ is constant, and depends thus only on the
trace~$x$. We denote it~$|x|_a$. For any infinite trace $\xi\in\BM$,
let $L_a(\xi)=\{x\in\M\tq x\leq\xi\wedge|x|_a>0\}$. If non empty, the
set $L_a(\xi)$ has a minimum which we denote by $V_a(\xi)$, and it
satisfies $|V_a(\xi)|_a=1$. Intuitively, $V_a(\xi)$~represents the
smallest sub-trace of $\xi$ with at least an occurrence of~$a$.

If $\BM$ is equipped with a Bernoulli measure~$\pr$, then
$L_a(\xi)\neq\emptyset$ with probability~$1$. Hence, neglecting a set
of zero probability, we may assume that $V_a:\BM\to\M$ is well
defined. The mapping $V_a$ is called the \emph{first hitting time
  of~$a$}. The \emph{distribution of the first hitting time of~$a$} is the
probability distribution of the random variable~$V_a$. It is a
discrete probability distribution on~$\M$, which we denote by~$\pr_a$,
and which is defined by $\pr_a(x)=\pr(V_a=x)$ for all $x\in\M$.

We will base our random generation of infinite heaps on the following
result.

\begin{theorem}
  \label{thr:10}
  Let\/ $\pr$ be a Bernoulli measure equipping the boundary $\BM$ of
  an irreducible trace monoid $\M=\M(\Sigma,I)$. Let $a\in\Sigma$, and
  let\/ $\pr_a$ be the distribution of the first hitting time of~$a$. Then\/
  $\pr_a^{*\infty}=\pr$, where\/ $\pr_a^{*\infty}$ is the limit distribution of
  the random walk on $\M$ associated with\/~$\pr_a$.
\end{theorem}

\begin{proof}[Sketch of proof]
  Let $(V^n)_{n\geq0}$ be the sequence of iterated stopping times
  associated with the first hitting time~$V_a$, as defined
  in~\cite[Def.~5.2]{abbes15}.  Under the probability~$\pr$, it
  follows from~\cite[Prop.~5.3]{abbes15} that, for each integer~$n$,
  $V^n$~has the same distribution as the $n^\text{th}$ step of the
  random walk associated with~$\pr_a$. Since $\M$ is assumed to be
  irreducible, the sequence $(V^n)_{n\geq0}$ is exhaustive as defined
  in~\cite[Def.~5.5]{abbes15}, from which the result derives.
\end{proof}

As a consequence, if we can simulate the distribution $\pr_a$ of the
first hitting time of some piece~$a$, we will be able to simulate a
$\pr$-distributed infinite random heap. The improvement lies in the
fact that first hitting times are finite heaps. Our next task consists
thus in studying more closely the distribution of the first hitting time, after
which we shall see how to simulate it.

\subsubsection{Pyramidal heaps and the distribution of the first hitting time}
\label{sec:pyramidal-heaps-1}

Recall that any trace has a interpretation as a labeled partial order
of pieces (see Section~\ref{sec:traces-as-labelled}). A trace $x$ is
\emph{pyramidal} if, as labeled partially ordered set, it has a
unique maximal element (a notion introduced by
Viennot~\cite{viennot86}). Any trace of the form $x=V_a(\xi)$ for some
$\xi\in\BM$ is pyramidal, with its unique occurrence of $a$ as its
unique maximal piece.

\begin{figure}
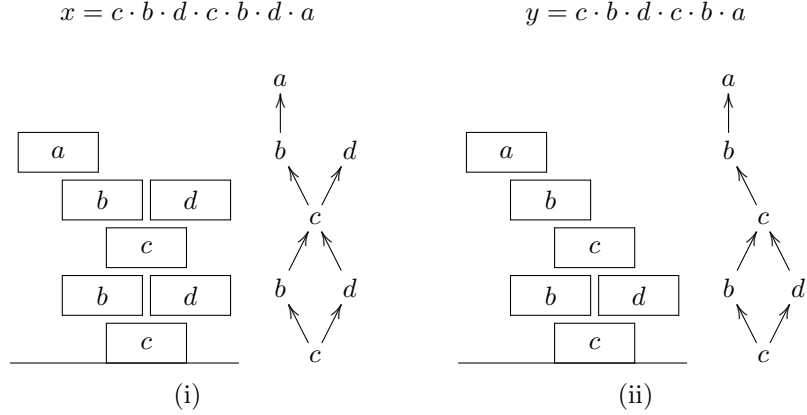

\begin{align*}
\begin{gathered}
 x=c\cdot b\cdot d\cdot c\cdot b\cdot d\cdot a\\[1em]
  \xy<.15em,0em>: (-2,0);(55,0)**@{-}, (0,48)="G",
  @={"G","G"+(20,0),"G"+(20,10),"G"+(0,10)}, s0="prev"
  @@{;"prev";**@{-}="prev"}, "G"+(10,5)*{a}, (22,0)="G",
  @={"G","G"+(20,0),"G"+(20,10),"G"+(0,10)}, s0="prev"
  @@{;"prev";**@{-}="prev"}, "G"+(10,5)*{c}, (11,12)="G",
  @={"G","G"+(20,0),"G"+(20,10),"G"+(0,10)}, s0="prev"
  @@{;"prev";**@{-}="prev"}, "G"+(10,5)*{b}, (33,12)="G",
  @={"G","G"+(20,0),"G"+(20,10),"G"+(0,10)}, s0="prev"
  @@{;"prev";**@{-}="prev"}, "G"+(10,5)*{d}, (22,24)="G",
  @={"G","G"+(20,0),"G"+(20,10),"G"+(0,10)}, s0="prev"
  @@{;"prev";**@{-}="prev"}, "G"+(10,5)*{c}, (11,36)="G",
  @={"G","G"+(20,0),"G"+(20,10),"G"+(0,10)}, s0="prev"
  @@{;"prev";**@{-}="prev"}, "G"+(10,5)*{b}, (33,36)="G",
  @={"G","G"+(20,0),"G"+(20,10),"G"+(0,10)}, s0="prev"
  @@{;"prev";**@{-}="prev"}, "G"+(10,5)*{d},
\endxy 
\quad \xy<.13em,0em>:
(10,0)*+{c};(20,20)*+{d}?>*\dir{>}**@{-},(20,20)*+{\phantom{d}};(10,40)*+{\phantom{c}}**@{-}?>*\dir{>},%
(10,0)*+{\phantom{c}};(0,20)*+{b}**@{-}?>*\dir{>};%
(0,20)*+{\phantom{b}};(10,40)*+{c}**@{-}?>*\dir{>},
(10,40)*+{\phantom{c}};(20,60)*+{d}**@{-}?>*\dir{>},
(0,60)*+{b}**@{-}?>*\dir{>},
(0,60)*+{\phantom{b}};(0,80)*+{a}**@{-}?>*\dir{>}
\endxy\\
\text{(i)}
\end{gathered}
&&
\begin{gathered}
  y=c\cdot b\cdot d\cdot
  c\cdot b\cdot a\strut\\[1em]
\xy<.15em,0em>:
(-2,0);(55,0)**@{-},
(0,48)="G",
@={"G","G"+(20,0),"G"+(20,10),"G"+(0,10)},
s0="prev" @@{;"prev";**@{-}="prev"}, "G"+(10,5)*{a},
(22,0)="G",
@={"G","G"+(20,0),"G"+(20,10),"G"+(0,10)},
s0="prev" @@{;"prev";**@{-}="prev"}, "G"+(10,5)*{c},
(11,12)="G",
@={"G","G"+(20,0),"G"+(20,10),"G"+(0,10)},
s0="prev" @@{;"prev";**@{-}="prev"}, "G"+(10,5)*{b},
(33,12)="G",
@={"G","G"+(20,0),"G"+(20,10),"G"+(0,10)},
s0="prev" @@{;"prev";**@{-}="prev"}, "G"+(10,5)*{d},
(22,24)="G",
@={"G","G"+(20,0),"G"+(20,10),"G"+(0,10)},
s0="prev" @@{;"prev";**@{-}="prev"}, "G"+(10,5)*{c},
(11,36)="G",
@={"G","G"+(20,0),"G"+(20,10),"G"+(0,10)},
s0="prev" @@{;"prev";**@{-}="prev"}, "G"+(10,5)*{b},
\endxy 
\quad \xy<.13em,.0em>:
(10,0)*+{c};(20,20)*+{d}?>*\dir{>}**@{-},(20,20)*+{\phantom{d}};(10,40)*+{\phantom{c}}**@{-}?>*\dir{>},%
(10,0)*+{\phantom{c}};(0,20)*+{b}**@{-}?>*\dir{>};%
(0,20)*+{\phantom{b}};(10,40)*+{c}**@{-}?>*\dir{>},
(10,40)*+{\phantom{c}};
(0,60)*+{b}**@{-}?>*\dir{>},
(0,60)*+{\phantom{b}};(0,80)*+{a}**@{-}?>*\dir{>}
\endxy\\
\text{(ii)}
  \end{gathered}
\end{align*}
\caption{\textsl{{\normalfont(i)}: For $\M=\langle a,b,c,d\tq ac=ca,\ ad=da,\
    bd=db\rangle$: a trace $x\notin\V_a$ (it is not pyramidal) and its
    associated labeled poset.\quad {\normalfont(ii)}: a trace $y\in\V_a$ and
    its associated labeled poset on the right.}}
\label{fig:poqpqq}
\end{figure}

Then the set $\V_a$ of traces $x\in\M$ in the image of the mapping
$V_a:\BM\to\M$ can be described as follows: $\V_a$~is the set of
pyramidal traces $x\in\M$ such that the piece $a$ only occurs as the
unique maximal piece; see Figure~\ref{fig:poqpqq}.  Furthermore, we
observe that, if $x\in\V_a$, then $\{V_a=x\}=\up x$ (an intuitive
property, also proved in~\cite[Prop.~4.2]{abbes15}). It follows that
the distribution $\pr_a$ of the first hitting time has the following
simple expression:
\begin{gather}
  \label{eq:4}
\forall x\in\M\quad\pr_a(x)=
\begin{cases}
0,&\text{if $x\notin\V_a$}\\
\pr(\up x),&\text{if $x\in\V_a$}  
\end{cases}
\end{gather}

\subsubsection{Generating pyramidal heaps}
\label{sec:gener-pyram-heaps}

We consider a network of alphabets $(\Sigma_1,\dots,\Sigma_N)$ with
$\Sigma=\Sigma_1\cup\dots\Sigma_N$, such that the synchronization
trace monoid $\M=\M(\Sigma,I)$ is irreducible. We pick an arbitrary
letter $a\in\Sigma$. For each $i\in\{1,\dots,N\}$, let
$\Sigma'_i=\Sigma_i\setminus\{a\}$ if $a\in\Sigma_i$, and
$\Sigma'_i=\Sigma_i$ otherwise. Then the synchronization monoid $\M'$
of $(\Sigma'_1,\dots,\Sigma'_N)$ coincides with the sub-monoid of $\M$
generated by $\Sigma\setminus\{a\}$.

We assume, for each $i\in\{1,\dots,N\}$, to be given~$p'_i$, either a
probability distribution or a sub-probability distribution
over~$\Sigma'_i$, such that the output $X$ of the \BSA\ executed on
$(\Sigma'_1,\dots,\Sigma'_N)$ with those parameters, is finite with
probability~$1$. 

Given the parameters $(p'_1,\dots,p'_N)$, we consider the execution of
the trial-and-reject procedure described in pseudo-code in
Algorithm~\ref{algo7} below. At each run of the loop, the algorithm
needs to decide whether some trace $V$ is pyramidal in $\M$ or not. It
is clear that a---far from being optimal---scanning procedure,
examining all elements of $V$ starting from the right, will
successfully complete this job.

\begin{algorithm}
\caption{Outputting a pyramidal trace in $\V_a$}
\label{algo7}
\begin{algorithmic}[1]
\Require--
\State $X\gets\unit$\Comment{Initialization}
\Repeat
\State\textbf{call} Algorithm~\ref{algo4}\Comment{Calling the \BSA}
\State $X\gets$ output of Algorithm~\ref{algo4}
\State $V\gets X\cdot a$
\Until $V$ is pyramidal in $\M$\Comment{See comment}
\State\Return $V$
\end{algorithmic}
\end{algorithm}

By assumption, the \BSA\ which is repeatedly called in
Algorithm~\ref{algo7} always terminates and outputs a finite
trace. Since the probability of success if positive, since for
instance $X=\unit$ yields a success and has positive probability to be
output by the \BSA, Algorithm~\ref{algo7} always terminates. The
probability distribution of the output is given in the following
lemma, stated with the notations introduced above.

\begin{lemma}
  \label{lem:3}
Let $f':\M'\to\bbR_+^*$ be the sub-M\"obius valuation associated with
the output of the\/ \BSA\ running on~$\M'$. Then the distribution of the output
$V$ of Algorithm~\ref{algo7} is concentrated on $\V_a$ and given by:
\begin{gather}
  \label{eq:24}
\forall v\in\V_a\quad \pr(V=v)=K\cdot f'(v/a),
\end{gather}
where $v/a$~denotes the heap obtained by removing from $v$ its unique
maximal piece~$a$, and $K$ is a normalization constant.
\end{lemma}

\begin{proof}
  Let $\prq$ denote the probability distribution of the output $X$ of
  the \BSA\ running on~$\M'$ with the specified parameters
  $(p'_1,\dots,p'_N)$. Then $f'(x)=\prq(\Up x)$ for
  $x\in\M'$. According to Theorem~\ref{thr:2}, for some constant
  $\varepsilon>0$, we have $\prq(X=x)=\varepsilon f'(x)$ for all
  $x\in\M'$.

  The rejection procedure amounts to considering the distribution of
  $V=X\cdot a$ conditioned on $X\cdot a\in\V_a$. Hence the probability
  for Algorithm~\ref{algo7} to issuing an element $v\in\V_a$ is: 
\begin{gather*}
 \pr(V=v)=\frac{\prq(X\cdot a=v)}{\prq(X\cdot
    a\in\V_a)}=\frac{\prq(X=v/a)}{\prq(X\cdot a\in\V_a)}=K\cdot f'(v/a),
\end{gather*}
where $K$ is the constant $K=\varepsilon/{\prq(X\cdot a\in\V_a)}$.
\end{proof}

Note that the form~(\ref{eq:24}) is almost that of a valuation
evaluated at~$v$. The contribution of the last piece~$a$ is missing,
but the constant $K$ is adequately placed to play the role of this
missing contribution. This will be used in the proof of
Theorem~\ref{thr:7}.

\subsection{The Probabilistic Full Synchronization Algorithm}
\label{sec:full-synchr-algor}

We are now ready for constructing a probabilistic algorithm generating
Bernoulli-distributed infinite traces. The framework consists of a
network $(\Sigma_1,\dots,\Sigma_N)$ of alphabets, such that the
synchronization trace monoid $\M=\M(\Sigma,I)$ is irreducible, with
$\Sigma=\Sigma_1\cup\dots\cup\Sigma_N$.

\subsubsection{Description of the algorithm}
\label{sec:prob-full-synchr}

Having chosen an arbitrary piece $a\in\Sigma$, we consider a family
$(p'_1,\dots,p'_N)$ of probabilistic parameters as above, \emph{i.e.},
with the constraint that the \BSA\ executed on the sub-monoid $\M'$
generated by $\Sigma\setminus\{a\}$ and with these parameters, outputs
a finite trace with probability~$1$. 

The Probabilistic Full Synchronization Algorithm (\FSA) is described
in pseudo-code in Algorithm~\ref{algo5} below. The \FSA\ is an endless
loop, incrementally writing to its output register~$X$. It simulates
thus the random walk on $\M$ with increments distributed according to
the distribution established in Lemma~\ref{lem:3}.

\begin{algorithm}
\caption{Probabilistic Full Synchronization Algorithm}
\label{algo5}
\begin{algorithmic}[1]
\Require--
\State $X\gets\unit$\Comment{Initialization}
\Repeat
\State\textbf{call} Algorithm~\ref{algo7}
\State $V\gets$ output of Algorithm~\ref{algo7}\Comment{Random pyramidal trace $V\in\V_a$}
\State $X\gets X\cdot V$\Comment{Increments the random walk}
\Until\textbf{false}
\end{algorithmic}
\end{algorithm}

The analysis of Algorithm~\ref{algo5} is twofold: a probabilistic
analysis carried on below and a complexity analysis carried on in
Section~\ref{sec:compl-analys-fsa}.

\subsubsection{Probabilistic analysis of the \FSA}
\label{sec:prob-analys-fsa}

Recall that, by convention, the \emph{output} of the \FSA\ is the
random infinite heap, least upper bound in $\BM$ of the finite heaps
recursively written in its output register. The probability
distribution of this infinite heap is as follows.

\begin{theorem}
  \label{thr:7}
  We consider the execution of the\/ \FSA\ in the framework described
  in Section~\ref{sec:prob-full-synchr}, and we adopt the same
  notations. Let $f':\M'\to\bbR_+^*$ be the sub-M\"obius valuation
  associated with the\/ \BSA\ executed with the chosen parameters\/
  $(p'_1,\dots,p'_N)$, and let $X_\infty$ be the output (with the
  convention recalled above) of the\/ \FSA.

  Then $X_\infty$ is distributed according to a Bernoulli measure
  on~$\BM$. The associated valuation $f:\M\to\bbR$ is the M\"obius
  valuation on $\M$ that extends $f':\M'\to\bbR$ (the existence and
  uniqueness of which are stated in Theorem~\ref{thr:4}).
\end{theorem}

\begin{proof}
  Let $f:\M\to\bbR_+^*$ be as in the statement. Let $\pr$ be the
  Bernoulli measure on $\BM$ defined by $\pr(\up x)=f(x)$, which is
  well defined according to Theorem~\ref{thr:2}.

  Letting $\prq$ be the distribution of~$X_\infty$, we have to prove
  that $\pr=\prq$. Let $g:\V_a\to\bbR_+^*$ be the probability
  distribution of the increment~$V$ in the \FSA.  Since $\prq$ is the
  limit distribution of the random walk on $\M$ with increments
  distributed identically to~$V$, it follows from Theorem~\ref{thr:10}
  that we only need to show that $g(x)=\pr(V_a=x)$ for all $x\in\V_a$.

  According to~(\ref{eq:4}), we have:
  \begin{align}
\label{eq:23}
\forall x\in\V_a\quad    \pr(V_a=x)&=\pr(\up x)=f(x).
  \end{align}
Whereas, according to Lemma~\ref{lem:3}, we have for some constant~$K$:
\begin{align}
\label{eq:25}
  \forall x\in\V_a\quad g(x)=K f'(x/a)=\frac{K}{f(a)}f(x).
\end{align}

Summing up over $\V_a$ in~(\ref{eq:23}) yields
$\sum_{x\in\V_a}f(x)=1$, whereas summing up over $\V_a$
in~(\ref{eq:25}) yields
$1=\bigl(K/f(a)\bigr)\cdot\bigl(\sum_{x\in\V_a}f(x)\bigr)=K/f(a)$. Therefore
$K=f(a)$, which yields after re-injecting in~(\ref{eq:25}): $g(x)=f(x)$
and thus $g(x)=\pr(V_a=x)$ for all $x\in\V_a$. Since $\pr$ has the
desired properties, the proof is complete.
\end{proof}

\subsubsection{Complexity analysis of the \FSA}
\label{sec:compl-analys-fsa}

We will limit our analysis to the following observation: \emph{the
  size of the output register of the \FSA\ grows linearly with time in
  average.}

Indeed, each \BSA\ call in Algorithm~\ref{algo7} takes a finite amount
of time in average according to Proposition~\ref{prop:3},
point~\ref{item:4}. Since the probability of success in the
trial-and-reject procedure of Algorithm~\ref{algo7} is positive, it
will thus execute in average in a fixed amount of time, whence the
average linear growth of the output register of the \FSA.

Another question is to compute adequately the probabilistic
parameters. We will discuss it briefly in
Section~\ref{sec:computational-issues}, after having examined some
examples.

\subsection{Example: ring models}
\label{sec:examples-1}

\subsubsection{A general result}
\label{sec:general-result}

For the ring models, the following result shows that any Bernoulli
measure can be simulated by executions of the \FSA. 

\begin{theorem}
  \label{thr:9}
Let $(a_0,\ldots,a_{N-1})$ be $N$ distinct symbols, and let\/
$(\Sigma_0,\ldots,\Sigma_{N-1})$ be the network of alphabets
defined by:
\begin{align*}
0<i\leq N-1\quad  \Sigma_i&=\{a_{i-1},a_i\}\,,&\Sigma_0&=\{a_{N-1},a_0\}\,.
\end{align*}

The synchronization monoid $\M=\M(\Sigma,I)$ is described as
follows:
\begin{align*}
\Sigma&=\{a_0,\ldots,a_{N-1}\}\,,&
I&=\{(a_i,a_j)\tq (i-j\mod N\geq 2)\wedge(j-i\mod N\geq2)\}\,.
\end{align*}

Then any Bernoulli measure on $\BM$ can be simulated by the endless
execution of the\/ \FSA, and in particular the uniform measure
on~$\BM$.
\end{theorem}

\begin{proof}
  Let $\pr$ be a target Bernoulli measure on~$\BM$, and let
  $f:\M\to\bbR_+^*$ be the associated M\"obius valuation. We pick
  $a_0$ as the piece to be removed. Let $\M'$ be the submonoid of
  $\M$ generated by $a_1,\dots,a_N$. Then $\M'$ is a path
  model. Furthermore, since $\M$ is irreducible, it follows from
  Theorem~\ref{thr:4} point~\ref{item:9} that the restriction of $f$ to
  $\M'$ is sub-M\"obius. According to Theorem~\ref{thr:8}, the
  associated probability distribution can be obtained by running the
  \BSA\ with suitable parameters. Running the \FSA\ based on this
  instance of the \BSA, we generate a Bernoulli measure which M\"obius
  valuation, say~$g$, extends to $\M$ the restriction of $f$
  to~$\M'$. By the uniqueness of the extension of sub-M\"obius
  valuations to M\"obius valuations (Theorem~\ref{thr:4}
  point~\ref{item:10}), it follows that $g=f$.
\end{proof}

\subsubsection{Example}
\label{sec:example}

We consider the example of the ring model $\M$ on five generators, the
synchronization graph of which is depicted in
Figure~\ref{fig:fiveoknq},~(i). We aim at generating the uniform measure
on~$\BM$, say~$\nu$, characterized by  $\nu(\up x)=p_0^{|x|}$\,, where
$p_0=\frac12-\frac{\sqrt5}{10}$ is the root of smallest modulus of
$\mu_\M(t)=1-5t+5t^2$\,. 

We pick $a_0$ as our distinguished piece, as depicted in
Figure~\ref{fig:fiveoknq},~(ii). Then, we wish to run on the sub-monoid
\begin{gather*}
\M'=\langle a_1,a_2,a_3,a_4\;|\;a_1a_3=a_3a_1,\ a_1a_4=a_4a_1,\ a_2a_4=a_4a_2\rangle
\end{gather*}
a \BSA\ with associated valuation $f':\M'\to\bbR$ such that
$f'(x)=p_0^{|x|}$ for all $x\in\M'$\,. Note that $p_0$ is \emph{not}
the root of~$\mu_{\M'}$\,!

\begin{figure}
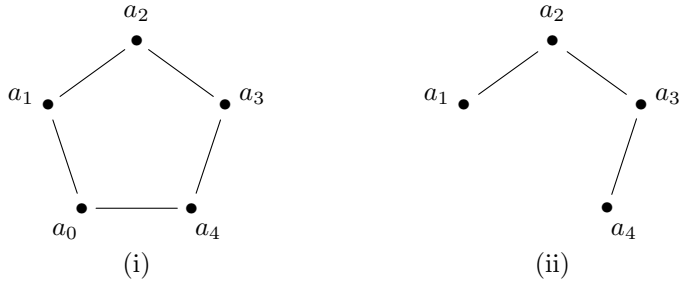

    \begin{align*}
\begin{gathered}
\xy{%
{\xypolygon5"A"{~:{(3.5,0):(3.5,0):}{\hbox{\strut$\;\bullet\;$}}}%
\xypolygon5"B"{~:{(4,0):(4,0):}~><{@{}}{}}}
,"B2",*{a_2},
,"B3",*{a_1},
,"B4",*{a_0},
,"B1",*{a_3},
,"B5",*{a_4}
}%
\endxy\\
\text{(i)}
\end{gathered}
&&
\begin{gathered}
\xy{%
{\xypolygon5"A"{~:{(3.5,0):(3.5,0):}~><{@{}}{}}%
\xypolygon5"B"{~:{(4,0):(4,0):}~><{@{}}{}}}
,"B1",*{a_3},
,"B5",*{a_4},
,"B2",*{a_2},
,"B3",*{a_1},
"A3",*+{\bullet};"A2"*+{\bullet}**@{-};"A1"*+{\bullet}**@{-};"A5"*+{\bullet}**@{-};
}%
\endxy\\
\text{(ii)}
\end{gathered}
\end{align*}
\caption{\textsl{{\normalfont(i)}:~synchronization monoid for the ring
    model with five generators.\quad {\normalfont(ii):~}generators of
    the path model on which the \BSA\ is run at each loop of the
    \FSA.}}
  \label{fig:fiveoknq}
\end{figure}

Referring to the computations performed in the proof of
Theorem~\ref{thr:8}, adequate solutions for $p_2$, $p_3$ and~$p_4$\,,
respectively on $\{a_1,a_2\}$, on $\{a_2,a_3\}$ and on
$\{a_3,a_4\}$\,, are obtained as follows:
\begin{align*}
  p_2(a_1)&=p_0=\frac12-\frac{\sqrt5}{10}\approx0.276
  &p_2(a_2)&=1-p_0=\frac12+\frac{\sqrt5}{10}\approx0.724\\
  p_3(a_2)&=\frac{p_0}{1-p_0}=\frac32-\frac{\sqrt5}2\approx0.382
  &p_3(a_3)&=\frac{1-2p_0}{1-p_0}=-\frac12+\frac{\sqrt5}2\approx0.618\\
  p_4(a_3)&=\frac{p_0(1-p_0)}{1-2p_0}=\frac1{\sqrt5}\approx0.447
  &p_4(a_4)&=p_0=\frac12-\frac{\sqrt5}{10}\approx0.276
\end{align*}

\section{Computational issues and perspectives}
\label{sec:computational-issues}

When trying to use the \FSA\ in practice for simulation and testing,
one might be concerned by the fact that it incrementally outputs heaps
with a particular shape, namely they are all pyramidal. Furthermore,
the tip of these pyramidal heaps is always labeled with the same
letter, corresponding to an arbitrary choice made before executing the
algorithm. Actually, there are good reasons not to worry about
that. Indeed, a large class of statistics on heaps can safely be
computed on these particular heaps, and they will asymptotically be
indistinguishable from statistics computed on arbitrary large
heaps. For instance, an asymptotics of the speedup, \emph{i.e.}\ the
ratio height over number of pieces in large heaps, can be estimated in
this way. Precise results on this topic are found in~\cite{abbes15},
under the name of cut-invariance.

Another concern is the following.  Even if the \FSA\ were proved to be
able to simulate any Bernoulli measure for any topology, not only for
the ring topology, there is no doubt that its execution would still
need the precomputation of adequate probabilistic parameters, and in
particular the root of smallest modulus of the M\"obius polynomial of
the synchronization monoid. Given that the determination of this
polynomial on a general synchronization graph is an NP-complete
problem (since the independence set decision problem is
NP-complete~\cite{dasgupta06:_algor}), this precomputation method
appears unrealistic to be used in practice.

However, what we need is not the M\"obius polynomial itself, but only
its root of smallest modulus. It is not theoretically forbidden to
think that this root might be approximated in polynomial time, even
though the M\"obius polynomial is hard to find. Actually, one can even
think of a feedback procedure based on our generation algorithms to
find an approximation of this root. Indeed, we could execute the
generation algorithm with arbitrary parameters, then adjust the
probabilistic parameters in order to increase uniformity, then re-run
the generation algorithm and re-adjust the parameters, ans so on. It
is reasonable to expect that such a procedure would lead
the parameters to converge toward the critical value entailing
uniformity, which is the root of smallest modulus of the M\"obius
polynomial. This interesting question may deserve a dedicated
work.

\bigskip
\noindent{\bfseries\sffamily Acknowledgments.}\quad
A number of ideas in this paper have emerged through animated
discussions with \'E.~Fabre, B.~Genest and N.~Bertrand at INRIA Rennes
during Spring 2015. I am grateful to A.~Muscholl for pointing out the
reference~\cite{cori85}; and I am grateful to C.~Male for pointing out
the reference~\cite{khanin10}. I am grateful to the anonymous reviewers
who greatly helped me improving the paper.

\printbibliography

\end{document}